\documentclass[10pt,reqno,tbtags]{amsart}

\usepackage{caption}
\usepackage{soul}

\usepackage{tikz}
\usetikzlibrary{calc,matrix,patterns}

\usepackage[top=30mm,bottom=40mm,inner=30mm,outer=35mm]{geometry}

\usepackage{amsmath,amssymb,amsfonts,amsthm} 
\usepackage{mathrsfs,latexsym} 
\usepackage{mathtools}
\usepackage{thmtools}

\numberwithin{equation}{section}

\DeclareMathAlphabet{\pazocal}{OMS}{zplm}{m}{n}

\usepackage{color}
\usepackage{cite}

\usepackage{epsfig}
\usepackage{graphicx}
\usepackage[all]{xy}

\usepackage{bm}
\usepackage{hyperref}

\DeclareFontFamily{OT1}{pzc}{}
\DeclareFontShape{OT1}{pzc}{m}{it}{<-> s * [1.1500] pzcmi7t}{}
\DeclareMathAlphabet{\mathpzc}{OT1}{pzc}{m}{it}

\usepackage{float}
\usepackage{wasysym}

\usepackage{cite}\medskip
\usepackage{todonotes}
\roman{enumi}

\newtheorem{theorem}{Theorem}[section]

\newtheorem{lemma}[theorem]{Lemma}
\newtheorem{proposition}[theorem]{Proposition}

\declaretheorem[sibling=theorem,style=definition, qed=\bell]{remark}

\numberwithin{equation}{section}

\newcommand{\CC}{{\mathbb C}}
\newcommand{\RR}{{\mathbb R}}

\newcommand{\NN}{{\mathbb N}}

\newcommand{\Dc}{{\mathscr{D}}}
\newcommand{\Ec}{{\mathscr{E}}}

\newcommand{\Fc}{{\mathscr{F}}}
\newcommand{\F}{{\mathcal{F}}}
\newcommand{\Gc}{{\mathcal{G}}}

\newcommand{\Lc}{{\mathcal{L}}}

\newcommand{\Rc}{{\mathscr{R}}}

\newcommand{\T}{\cdot_T}

\newcommand{\fG}{{\mathfrak G}}

\newcommand{\X}{{\mathcal X}}
\newcommand{\Y}{{\mathcal Y}}

\newcommand{\loc}{\rm loc}

\newcommand{\Floc}{\Fc_{\mathrm{loc}}}                   

\newcommand{\supp}{{\mathrm{supp} \, }}

\newcommand{\no}[1]{:\! \! #1 \! \!:}

\newcommand{\be}{\begin{equation}}
\newcommand{\ee}{\end{equation}}

\newcommand{\Lap}{\triangle}

\def\eg{{\it e.g.\ }}
\def\ie{{\it i.e.\ }}

\newcommand{\0}{\emptyset}
\newcommand{\g}{\mathpzc{g}}

\newcommand{\h}{\mathpzc{h}}

\newcommand{\e}{\mathpzc{e}}
\newcommand{\G}{\mathpzc{G}}

\newcommand{\LieGc}{\mathrm{Lie}\G_c(M)}
\newcommand{\LieRc}{\mathrm{Lie}\,\Rc_c}
\newcommand{\Fmloc}{\Fc_{\bullet\,\mathrm{loc}}}  
\newcommand{\Fnloc}[1]{\Fc_{#1\,\mathrm{loc}}}   

\newcommand{\BV}{\mathrm{BV}}
\newcommand{\BVloc}{\mathrm{BV}_{\mathrm{loc}}}
\newcommand{\BVmloc}{\mathrm{BV}_{\bullet\,\mathrm{loc}}}
\newcommand{\BVnloc}[1]{\mathrm{BV}_{#1\,\mathrm{loc}}}

\begin{document} 

\title[The UAMWI and the Wess-Zumino consistency condition]{Unitary, anomalous Master Ward Identity and its connections to the 
Wess-Zumino condition, BV formalism and $L_\infty$-algebras}

\author[R. Brunetti]{Romeo Brunetti}
\address{Dipartimento di Matematica, Universit\`a di Trento, 38123 Povo (TN), Italy}
\email{romeo.brunetti@unitn.it}

\author[M. D\"utsch]{Michael D\"utsch}
\address{Institute f\"ur Theoretische Physik, Universit\"at G\"ottingen, 37077 G\"ottingen, Germany}
\email{michael.duetsch3@gmail.com}

\author[K. Fredenhagen]{Klaus Fredenhagen}
\address{II. Institute f\"ur Theoretische Physik, Universit\"at Hamburg, 22761 Hamburg, Germany}
\email{klaus.fredenhagen@desy.de}

\author[K. Rejzner]{Kasia Rejzner}
\address{Department of Mathematics, University of York, YO10 5DD York, UK}
\email{kasia.rejzner@york.ac.uk}

\maketitle

\begin{abstract}
The C*-algebraic construction of QFT by Buchholz and one of us relies on the causal structure of spacetime and a classical Lagrangian. 
In one of our previous papers we have introduced additional structure into this construction, namely an action of symmetries, which is related to fixing renormalisation conditions. This action characterizes anomalies and 
satisfies a cocycle condition which is summarized in the unitary anomalous Master Ward identity. Here (using perturbation theory) we
show how this cocycle condition is related to the Wess-Zumino consistency relation and the consistency relation for the anomaly in the BV formalism, 
where the latter follows from the generalized Jacobi identity for the associated $L_\infty$-algebra. 
In addition, we give a proof that perturbative agreement (i.e., 
independence of a perturbative QFT on the splitting of the Lagrangian into free and interacting parts) can be achieved by finite renormalizations.
\end{abstract}
\section{Introduction}

One of the most interesting features of quantum physics is the fact that symmetries of the classical theory are, in general, not straightforwardly transferred to the corresponding quantum theory. Instead, often the symmetries are modified by \emph{anomalies}. These satisfy the \emph{Wess-Zumino consistency relations} \cite{WZ71}, and the arising new structures have a crucial impact on the quantum theory, e.g. on the formulation of the standard model of particle physics.

In perturbative algebraic quantum field theory (pAQFT), the anomalies can be obtained in terms of the anomalous Master Ward Identity (AMWI) \cite{Brennecke08,DB02,DF03}, and it was shown by Hollands \cite{Hollands08} that in Yang-Mills theory these anomalies satisfy a consistency relation which allows one to apply the homological methods of the BRST-BV formalism, 
where the key information about the theory is encoded in a certain differential. In \cite{FredenhagenR13,Rej13} this result was generalized to arbitrary theories with local gauge symmetries 
and an infinite-dimensional rigorous version
of the Batalin-Vilkovisky (BV) formalism was formulated, without the use of path integral or any regularization scheme (in contrast to \cite{Costello}).
One of the crucial results of that work was to describe the difference between classical symmetries and their quantized counterparts in terms of deformation of the classical BV differential to the quantum BV differential. This deformation is induced by the deformation of the pointwise product of the classical theory to the renormalized time-ordered product. 
In particular, a renormalized BV Laplacian was introduced, and its action on the BV algebra could be understood in terms of the anomaly \cite{Rej13}. Recently, Fr\"ob \cite{Fro} succeeded in
proving that the arising algebraic structure is that of an $L_{\infty}$-algebra. The Wess-Zumino consistency relation has also been applied recently in \cite{SZ17} in the treatment of global anomalies. Another insight concerns the difference between consistent anomalies, \ie those that satisfy the Wess-Zumino conditions, and the so-called covariant anomalies \cite{BZ84}. We do not enter into this in our paper and refer the reader to the literature \cite{Bert}.

In a previous paper \cite{BDFR21}, we investigated the action of symmetries in the C*-algebraic construction of scalar quantum field theories proposed in \cite{BF19}. 
In that construction the algebras are generated by S-matrices which describe local interactions within  compact regions of spacetime. Subject to a causality condition and a unitary version of the Schwinger-Dyson equation, one obtains a net of C*-algebras satisfying the Haag-Kastler axioms, generalized to generic globally hyperbolic spacetimes according to the principles of locally covariant QFT \cite{BFV03}. Starting from the free Lagrangian and admitting only linear interactions, one obtains the well known Weyl algebra of the free field.
If one includes more general interactions, the arising algebra possesses automorphisms which act nontrivially only in a compact subregion. The existence of such internal symmetries violates the time slice axiom which states that observations in the neighborhood of some Cauchy surface determine all other observables, \ie the algebra associated to this neighborhood is already the algebra of the whole spacetime. 

Therefore, we introduced in \cite{BDFR21} an
additional axiom for these C*-algebras:
the ``unitary anomalous Master Ward Identity (UAMWI).'' 
It characterizes how \textit{symmetries of the classical configuration space}, which are not necessarily symmetries of the Lagrangian, are modified in the quantum theory. The symmetries considered form a group $\G_c$ of transformations with compact support,
generated by affine field redefinitions
and point transformations  \footnote{We remark that none of the considered Lagrangians is invariant under such transformations.}.
Their action on observables is modified by a map (the anomaly term) $\zeta$ from $\G_c$
to a group $\Rc_c$
of transformations of functionals.
In perturbation theory, the elements of $\Rc_c$ relate different choices of time ordering prescriptions which not necessarily satisfy covariance conditions. According 
to Stora's \emph{Main Theorem of Renormalization} \cite{PopineauS16,DF04}
these transformations form a subgroup of the St\"uckelberg-Petermann renormalization group \cite{BDF09}. This subgroup is determined by the conditions Action Ward Identity, Field Equation and Field Independence on the time ordering and the requirement of compact support. It should be distinguished from other subgroups related to covariance conditions. Actually, the intersection with such subgroups is typically trivial.
Most importantly for this paper, $\zeta$ satisfies a cocycle relation. We also showed that this cocycle $\zeta$ exists 
in the perturbative version of the model where it can be determined up to equivalence and yields the known anomalies. 

Building on these results, in the present paper we explore further the cocycle condition, focusing on perturbation theory. Since the non-triviality of this cocycle is related
to the existence of anomalies, it is reasonable to expect that it should be related to the Wess-Zumino consistency condition. The latter has been originally derived in the context of the effective action and it reflects to the way in which this  action transforms under infinitesimal gauge symmetries. In Section~\ref{sec:WZ}  we review that original derivation, following essentially \cite{AG85}. Although it is clear that the Wess-Zumino consistency condition has to be related to the action of the Lie algebra of the group of symmetries of the theory, the precise statement of this fact in the framework of \cite{BF19,BDFR21} has not been known. While addressing this question, the present work also makes connections with another statement of the Wess-Zumino consistency condition, namely the one present in the BV formalism.

Concretely, we show that, considering the infinitesimal  symmetry transformations,
the cocycle $\zeta$ induces a corresponding map $\Delta:\mathrm{Lie}\G_c\to\LieRc$ which is a Lie algebraic cocycle, 
and that this cocycle is the anomaly map appearing in the AMWI
(Theorem 10.3 in \cite{BDFR21} and Theorem \ref{thm:Liecc} in this paper). This provides a link between the notions of anomalies used in perturbation theory \cite{Brennecke08} and anomalies in the non-perturbative formulation of \cite{BDFR21}.
In  Section~\ref{sec:infinitesimal-cc}, in Theorem~\ref{thm:cc}, we give another derivation of the cocycle relation for $\Delta$: we show
that the anomaly $\Delta$ of the AMWI satisfies a consistency condition, which is precisely the cocycle relation for $\Delta$, and which
we call the \textit{extended Wess-Zumino consistency condition}, as it reduces to the standard Wess-Zumino condition for quadratic interactions.

Finally, we discuss the relation to the BV formalism. In \cite{FredenhagenR13}, two of us have shown that the anomaly in the AMWI is in fact related to the \textit{renormalised BV Laplacian}, so it is natural to expect that the algebraic properties of the BV Laplacian would be reflected also in the cocycle condition. This is indeed the case, as we prove in Section~\ref{sec:BV} that the extended cocycle condition for $\Delta$ follows directly from the nilpotency of the BV operator, when applied to those infinitesimal symmetries which arise from affine field 
redefinitions $\g\in\G_c$ (Prop.~\ref{prop:BV->WZ}).

There is another structure which is often studied in perturbation theory, namely the principle of perturbative agreement 
\cite{HW04,Zahn15,DHP17}, which requires that the perturbative construction should not depend on the way the Lagrangian is splitted into free and interacting part. We prove that this principle can be satisfied by finite renormalizations and clarify its relation to the UAMWI (subsection \ref{PA}).

\section{The framework}\label{sec:properties}

\subsection{Perturbative algebraic quantum field theory (pAQFT)} We use the same setting for pAQFT as in \cite[Sect.~10 and App.~C]{BDFR21}; for the convenience of the reader we repeat here in a somewhat sketchy way the 
notations, definitions and results being relevant for this paper.

We consider an $n$-component real scalar field $\Phi$ on a globally hyperbolic curved space-time $M$ of dimension larger than 2. 
The classical configuration space $\Ec(M,\RR^n)$ is the space of 
smooth functions on $M$ with values in $\RR^n$. 
The basic field $\Phi(x)$ is the evaluation functional
\be\label{eq:basic-field}
\Phi(x)\,:\,\Ec(M,\RR^n)\to\RR^n\,; \,\Phi(x)[\phi]=\phi(x)\ .
\ee
Observables are elements of the space $\Fc(M)$ of functionals $F:\Ec(M,\RR^n)\to\CC$ which are polynomial in $\phi$ and have the form
\be 
F[\phi]=\sum_{k=0}^m\langle f_k,\phi^{\otimes k}\rangle
\ee
with compactly supported distributional densities $f_k$ on $M^k$ satisfying suitable conditions on their wave front sets \cite{BF00,BDF09}.
The latter ensures the existence of the star product 
of the free theory (which is given in terms of the free Lagrangian $L$, see below) as a map
$\star:\Fc(M)\times\Fc(M)\to\Fc(M)$. This star product is an $\hbar$-dependent deformation of the (commutative) pointwise product: 
$F\cdot G[\phi]\doteq F[\phi]G[\phi]$ for $F,G\in\Fc$, $\phi\in\Ec(M,\RR^n)$, see \cite{DF01} or \cite[Chap.~2]{D19}. The (functional) support of a functional $F$ as above is the smallest closed set $K\subset M$ such that $\supp f_k\subset K^k$ for all $k$ (where $\supp f_0=\0$ is understood).

The subspace of \emph{local} functionals $F\in\Floc(M)$ is defined by the additional conditions that $F$ is \emph{$\RR$-valued} and of the form
$F[\phi]=\int \hat{F}(x,j_x(\phi))$, with a smooth density-valued function $\hat{F}$ on the jet space of $\Ec(M,\RR^n)$ with compact 
support in $x$. 

The Lagrangian $L$ is the usual Lagrangian of the free theory,
$$
L(x)[\phi]=\frac12 \bigl(g^{-1}(d\phi(x),d\phi(x))+m^2(\phi(x),\phi(x))\bigr)\,d\mu_g(x)\ ,
$$
where we use the canonical metric on $\RR^n$; $L(x)$ is a density with values in the local functionals and we write  
$L(f)\doteq\int_M L(x)f(x)\in\Floc(M)$ for $f\in\Dc(M,\RR)$.
The $\star$-product is given in terms of a Hadamard function $H$, \ie a bisolution of the associated Euler-Lagrange operator, the Klein-Gordon operator $K$,
\begin{equation}\label{eq:star}
    F\star G[\phi]=e^{\langle\frac{\delta}{\delta \phi},H\frac{\delta}{\delta\phi'}\rangle}F[\phi]G[\phi']|_{\phi'=\phi}\ .
\end{equation}
$H$ is of positive type, its antisymmetric part is $\frac{i}{2}(\Delta^{\mathrm{R}}-\Delta^{\mathrm{A}})$ with the retarded (R) and advanced (A) Green operator of $K$, 
and its wave front set satisfies the microlocal spectrum condition \cite{Rad96,BrunettiFK96}. 
There is no unique Hadamard function, but different choices differ by smooth bisolutions and lead to equivalent $\star$-products.

To construct the time ordered product we use an off-shell
version of the Epstein-Glaser method\footnote{Epstein and Glaser consider Fock space operators of the form $\sum_k \langle f_k,\no{\varphi^{\otimes k}}\rangle$ with the normal ordered products of the free field $\varphi$. This corresponds to a restriction of functionals to the space of solutions of the free field equation (on shell formalism).} 
\cite{EG73}, generalized to globally hyperbolic space times \cite{BF00,HW01,HW02}.
We furthermore use the fact that the $k$-fold pointwise product of local functionals which vanish at the zero configuration is injective and thus isomorphic to its image, the $k$-local functionals%
\footnote{Note that $\Floc(M)=\Fnloc{0}(M)+\Fnloc{1}(M)$.}
$F\in\Fnloc{k}(M)$. Identifying the $1$-fold product with the identity and the $0$-fold product with the map $\RR\ni c\mapsto F_c$ with the constant functional $F_c[\phi]=c$, we can describe the time ordered product as a linear map 
\be\label{eq:T-product}
T:\Fmloc(M)\to\Fc(M)
\ee
where $\Fmloc(M)$, the space of multilocal functionals, is the direct sum of the spaces $\Fnloc{k}(M)$ of $k$-local functionals, $k\in\NN_0$ \cite{FredenhagenR13}.
We can then equip the space $T\Fmloc(M)$ with the commutative and associative product
\be\label{eq:cdotT}
F\cdot_T G\doteq T((T^{-1}F)\cdot (T^{-1}G))\ .
\ee

On local functionals $T$ is the identity. In the sense of formal power series we can then characterize $T$ by its action on exponentials of local functionals $S(F)=T e^{iF}\equiv e_{\cdot_T}^{iF}$, the \emph{formal S-matrices}. They are unitaries with respect to the $\star$-product and 
have to satisfy the condition of \emph{causal factorization}
\be\label{eq:caus}
S(F+G)=S(F)\star S(G)\quad\text{if $\supp F\cap J_-(\supp G)=\emptyset$}
\ee
where $J_-$ denotes the causal past of a space-time region and $F,G\in\Floc(M)$. 
The time-ordered product is further restricted by renormalization conditions: as explained in the introduction, we do not impose 
any  covariance conditions; besides the Action Ward Identity which is implicit in our formalism since the time ordered products depend only on the functionals but not on the way they are obtained as integrals over functions of the fields, 
we only require field independence
\be\frac{\delta}{\delta\phi}T(F)=T\left(\frac{\delta}{\delta\phi}F\right),\quad F\in\Fmloc(M),
\ee
and the off-shell field equation%
\footnote{In \cite{HW04} this renormalization condition is called ``free field factor axiom."}
\be\label{eq:T-field-eq}
T\bigl(F\cdot\langle\Phi,f\rangle\bigr)=T\bigl(\langle F',E^{\mathrm{F}}f\rangle\bigr)+T(F)\cdot\langle\Phi,f\rangle\ ,\ f\in\Dc(M,\RR^n),\quad F\in\Fmloc(M),
\ee
where $E^{\mathrm{F}}=H+i\Delta^{\mathrm{R}}$ is the Feynman propagator associated to the Hadamard function $H$ and $F'$ is the first derivative of $F$. Time ordered products satisfying these axioms exist, and, according to Stora's {Main Theorem of Renormalization} 
\cite{PopineauS16,DF04}, any two formal S-matrices $S$ and $\hat{S}$ are related by
\be\label{eq:MainT}
  \hat S=S\circ Z
\ee
where $Z$ is a formal power series 
\be\label{eq:Z-perturbative}
  Z(F+c)=c+\sum_{n=1}^{\infty}\frac{1}{n!}Z_n(F^n)
\ee
with linear maps $Z_n:\Fnloc{n}(M)\to\Floc(M)$, $F\in\Fnloc{1}(M), c\in\RR$. The \emph{St\"uckelberg-Petermann renormalization group} $\Rc_0$ is defined 
to be the set of maps $Z\equiv(Z_n)_{n\in\NN}$ appearing in \eqref{eq:MainT}, and one proves that this set is indeed a group \cite{DF04}.
For a direct definition of $\Rc_0$ see \cite{DF04,BDF09} or \cite[Chap.~3.6]{D19}. We recall that $Z$ commutes with the addition of constant functionals, $Z(F+c)=Z(F)+c$, and that it acts on quadratic functionals $V$ by adding a constant: $Z(V)=V+c$ for some $c\in\CC$ depending on $V$.

We immediately see that $Z(0)=0$ and $Z_1=\mathrm{id}$. To include also possible changes of the Feynman propagator which is unavoidable
in a generally covariant formalism \cite{HW02},
we generalize the definition of $\Rc_0$ by admitting nontrivial, but still invertible $Z_1$ which describe the change of the normal ordering and thus the action of the time ordering operator $\hat{T}$ on 1-local functionals. For convenience, we continue to use a time ordering operator $T$ which is the identity on local functionals and
obtain the more general time orderings by composition with the renormalization group map $Z$ as in equation \eqref{eq:MainT}.

We do not add conditions on covariance to our renormalization conditions, since we want to have the freedom to add quite general external fields to our system. Covariance under certain symmetries then becomes visible in the triviality of the corresponding cocycles.
\subsection{Unitary anomalous master Ward identity}
In pAQFT, the unitary anomalous master Ward identity (UAMWI) describes the behaviour of the time-ordered product under the group $\G_c(M)$ of compactly supported automorphisms of the affine bundle $M\times \RR^n$. This group is generated by the following transformations $\g:\Ec(M,\RR^n)\to\Ec(M,\RR^n)$:
\begin{itemize}
   
    \item \emph{Point transformations}, \ie   smooth and compactly supported diffeomorphisms $\rho:M\to M$
    inducing the transformation $\g_\rho:\phi\mapsto \g_\rho(\phi)\doteq\phi\circ\rho$.
    
    \item \emph{Affine field redefinitions} $\g_{(A,\psi)}$ 
    with $A\in\Dc(M,\mathrm{GL}(n,\RR))$ and $\psi\in\Dc(M,\RR^n)$ which act on configurations by
    \be\label{eq:g(affine)-phi}
    \g_{(A,\psi)}(\phi)(x)\doteq \phi(x)A(x)+\psi(x)\ .
    \ee
    where $\phi(x)$ and $\psi(x)$ are considered as row vectors.
    
\end{itemize}
The action of $\G_c(M)$ on a functional $F\in\Floc(M)$ is defined by 
\be\label{eq:g-star}
   \g_\ast F[\phi]\doteq F[\g(\phi)]
\ee
and the free Lagrangian $L$ is transformed by
\be\label{eq:g_*L}
  ((\g_\rho)_\ast L)(f)\doteq(\g_\rho)_{\ast}(L(f\circ\rho))\ ,\quad ((\g_{(A,\psi)})_\ast L)(f)\doteq (\g_{(A,\psi)})_\ast(L(f))
\ee
with $f\in\Dc(M,\RR)$. Note that $(\g\h)_\ast=\g_\ast\h_\ast$ for $\g,\h\in\G_c(M)$, that is, $\G_c(M)\ni\g\mapsto\g_\ast$ is a 
representation of $\G_c(M)$ by maps on $\Floc(M)$.

The group $\G_c(M)$ acts on the full Lagrangian, and hence on the interaction by an $L$-dependent action  on $\Floc(M)$
\begin{equation}\label{eq:g_L}
    (\g,F)\mapsto \g_LF\doteq \delta_{\g} L+\g_{\ast } F\quad\text{where}\quad\delta_\g L\doteq \g_\ast L(f)-L(f)
\end{equation}
with $f\in\Dc(M,\RR)$ such that $f\vert_{\supp\g}=1$.
Obviously $\mathpzc{e}_L=\mathrm{id}_{\Floc(M,L)}$ for the unit $\e\in \G_c(M)$, and one verifies that  $(\g\h)_L=\g_L\circ\h_L$.

The unitary anomalous master Ward identity (UAMWI) relates the transformations induced by the action $\g\to\g_L$ to renormalization group transformations $\zeta_{\g}\in\Rc_0$ with $\supp\zeta_{\g}=\supp\g$. Here the support of $Z\in\Rc_0$ is the smallest closed subset 
$N$ of $M$ such that $Z(F+G)=F+Z(G)$ for all $F,G\in\Floc(M)$ with $\supp F\cap N=\emptyset$. The subgroup $\Rc_c$ of $\Rc_0$ of renormalization group maps $Z$ with \emph{compact support} was discussed in \cite[Appendix~C]{BDFR21}.

We will now discuss the UAMWI in pAQFT. To allow for an off-shell description, we introduce ``sources.''
Let $q\in\Ec_{\mathrm{dens}}(M,\RR^n)$ be a  smooth density and we define $L_q\doteq L-\langle\Phi,q\rangle$. In pAQFT, the UAMWI states that there exists a map (called ``anomaly map")
\be
\zeta:\mathpzc{G}_c(M)\rightarrow \Rc_c,\quad\text{satisfying} \quad\zeta_\e=\mathrm{id}_{\Floc(M)},\quad \supp\zeta_{\g}\subset\supp \g
\ee
and the \emph{cocycle relation}%
\footnote{Note that all the maps in the cocycle relation are nonlinear, and we use juxtaposition to denote composition of maps. For the composition of a renormalization group transformation $Z$ with $\g_L$ this means in terms of the linear maps $Z_n$ for $F\in\Fnloc{1}(M)$
$$Z\circ\g_L(F)=\g_L(F)[0]+\sum_n\frac{1}{n!}Z_n\bigl((\g_L(F)-\g_L(F)[0])^n\bigr)\ .$$}
\begin{equation}\label{eq:cocycle}
 \zeta_{\g\h}=\zeta_{\h}\,\h_L^{-1}\zeta_{\g} \h_L\ ,\quad \g,\h\in\G_c(M)\ ,
\end{equation}
such that for every smooth density $q\in\Ec_{\mathrm{dens}}(M,\RR^n)$
\begin{equation}\label{eq:uni-anom-MWI}
    S\circ\g_{L_q}(F)[\phi]= S\circ\zeta_\g(F)[\phi]\ , \text{ for $\phi$ solving}\ \frac{\delta L}{\delta\phi}[\phi]=q \ ,
\end{equation} 
with $\g\in\G_c(M),\,\,F\in\Floc(M)$ arbitrary.
As shown in Theorem 10.3 in \cite{BDFR21}, the UAMWI follows%
\footnote{In \cite{BDFR21} only the case $q=0$ was treated. The generalization to arbitrary densities $q$ relies on the fact that $\zeta$ does not change under adding a source term $-\langle\Phi,q\rangle$ to the Lagrangian, see Theorem \ref{th:AMWI-uAMWI} in appendix \ref{app:AMWI-uAMWI}.
In particular, the proof of that Theorem explicitly shows that the cocycle relation \eqref{eq:cocycle} is a necessary condition for the UAMWI \eqref{eq:uni-anom-MWI}.}
from the anomalous master Ward identity (AMWI) \cite{Brennecke08} (recalled below in \eqref{eq:AMWI-q} or \eqref{eq:anomMWI}), which is its infinitesimal version, formulated in terms of the respective Lie algebras.

The Lie algebra $\LieRc$ is defined as follows (compare \cite[Appendix~C]{BDFR21}): it is
the space of formal power series $z(F)=\sum_{n=1}^\infty \frac{1}{n!}z_n(F^n)$, with linear maps $z_n:\Fnloc{n}(M)\to\Floc(M)$, with 
the properties
\begin{itemize}
    \item[$(\mathrm{P}1)$]$\mathrm{id}+\lambda z_1$ is invertible for $\lambda$ sufficiently small,
    \item[$(\mathrm{P}2)$] $z(F+G)=z(F)+z(G)$ for $\supp F\cap\supp G=\0$, $F,G\in\Floc(M)$,
    \item[$(\mathrm{P}3)$]$z(F+\langle\Phi,\psi\rangle)=z(F)$ for $\psi\in\Dc(M,\RR^n)$, 
     \item[$(\mathrm{P}4)$]$\frac{\delta}{\delta\phi}z(F)=\langle z'(F),(\frac{\delta}{\delta\phi}F)\rangle$ ,
     \item[$(\mathrm{P}5)$] the support of $z$ is compact, where $\supp z$ is the smallest closed subset $N$ of $M$ such that $z(F+G)=z(G)$
for all $F,G\in\Floc(M)$ with $\supp F\cap N=\emptyset$.
\end{itemize}

The action of $\G_c(M)$ on the configuration space (considered as an affine space) induces an action of the Lie algebra $\LieGc$
with values in the associated vector space,
\be
  \Ec(M,\RR^n)\times\LieGc\ni(\phi,X)\mapsto \phi X \ .
\ee
To determine the Lie bracket,
it is convenient to describe the Lie algebra $\LieGc$ in a faithful representation of the group. Since $\G_c(M)$ acts from the right on field configurations, 
we write it in terms of a matrix multiplication from the right on the space $\Ec(M,\RR^n)\oplus\RR$ in the form
\begin{equation}\label{eq:repres-of-G}
 \g\,:\,   (\phi,c)\mapsto (\phi,c)\left(\begin{array}{cc}
       A  &  0\\
        \psi & 1
    \end{array}\right)\circ\left(\begin{array}{cc}
      \rho  &  0\\
        0 &  \mathrm{id} 
    \end{array}\right)\ ,
\end{equation}
with $\g=(A,\psi,\rho)$,
from which we get
\begin{equation}\label{eq:repres-of-LieG}
X\,:\, \phi\mapsto  (\phi,1)\left(\begin{array}{cc}
       a+\stackrel{\leftarrow}{\partial}_{\mu} v^{\mu} & 0 \\
        p & 0
    \end{array}\right)=(\phi a+v^{\mu}\partial_{\mu}\phi+p,0)\doteq(\phi X,0)
\end{equation}
with the Lie algebra element $X=(a,p,v)$, $a\in\Dc(M,\mathrm{gl}(n,\RR))$, $p\in\Dc(M,\RR^n)$ and a smooth vector field $v$ with compact support. The Lie bracket 
\be
[(a,p,v),(b,q,w)]=([a,b]+w^{\nu}\partial_{\nu} a-v^{\nu}\partial_{\nu} b,pb-qa+w^{\nu}\partial_{\nu}p-v^{\nu}\partial_{\nu}q,w^{\nu}\partial_{\nu}v-v^{\nu}\partial_{\nu}w) 
\ee
can directly be obtained from the matrix representation above. Note the unusual sign of the Lie bracket of vector fields due to the action of derivatives to the functions on the left, indicated by the upper left arrow. 

The action of $\LieGc$ on field configurations yields a representation $X\mapsto \partial_X$ on the space of local functionals 
with
\be
\partial_XF[\phi]=\langle F'[\phi],\phi X\rangle\equiv\int \frac{\delta F}{\delta\phi_a(x)}[\phi](\phi X)_a(x)\,,
\ee
where we use the usual summation conventions over the components of $\phi$, $a=1,\dots,n$; and
the functional derivative is naturally identified with a density. In particular we have
\be\label{eq:bracket}
[\partial_X,\partial_Y]=\partial_{[X,Y]}\ .
\ee
For the Lagrangian, we set
\be
\partial_XL\doteq\partial_XL(f)\quad\text{with $f\in\Dc(M,\RR)$ satisfying $f\vert_{\supp X}=1$.}
\ee

To formulate explicitly the above statement, that the AMWI is the infinitesimal version of the UAMWI,  
let $X\in\LieGc$ be the tangent vector at $\lambda=0$ of a smooth curve $\lambda\mapsto \g^\lambda\in\G_c(M)$ with $\g^0=\e$. 
Starting with the UAMWI \eqref{eq:uni-anom-MWI} we substitute $\g^\lambda$ for $\g$ and apply $\frac{d}{d\lambda}\vert_{\lambda=0}$; 
this yields 
\be\label{eq:AMWI-q}
T\Bigl(e^{iF}\cdot\bigl(\partial_X F+\partial_X L_q -\Delta X(F)\bigr)\Bigr)[\phi]=0\ , \text{ for }\phi
\text{  solving }\ \frac{\delta L}{\delta\phi}[\phi]=q\,\ 
\ee
with $X\in\LieGc,F\in\Floc(M)$ and where 
\be\label{eq:Delta-X}
\Delta\,:\,\LieGc\ni X\mapsto\Delta X\doteq\frac{d}{d\lambda}\Big\vert_{\lambda=0}\zeta_{\g^\lambda}\in\LieRc\,.
\ee
Indeed, \eqref{eq:AMWI-q} agrees with the AMWI, thus 
$\Delta X(F)$ coincides with the uniquely determined anomaly in the AMWI (see \cite[Thm.~7]{Brennecke08}, \cite[Thm.~5.2]{BrenneckeD09} and \cite[Chap.~4.3]{D19}). The AMWI \eqref{eq:AMWI-q} may also be written in the equivalent form
\be\label{eq:anomMWI}
e_T^{iF}\cdot_T\bigl(\partial_X F+\partial_X L-\Delta X(F)\bigr)=\int \ \bigl(e_T^{iF}\cdot_T (\partial_X\Phi(x))_a\bigr)
\frac{\delta L(f)}{\delta \phi_a(x)}\,,\ f\equiv 1 \text{ on }\supp X\ , 
\ee
where $\delta L(f)/\delta \phi_a(x)$ is understood as a density.
We observe that the map $X\mapsto\Delta X$ is linear and that $\supp\zeta_{\g}\subset\supp\g$ implies $\supp \Delta X\subset\supp X$.
Moreover, there is a common locality of $\Delta X(F)$ in $X$ and $F$ derived in \cite[Thm.~7]{Brennecke08}, see also \cite[Thm.~4.3.1]{D19}:
\begin{lemma}\label{lem:prop}
The anomaly map $\Delta$ of the AMWI satisfies
\be
\supp \Delta X(F)\subset \supp F\cap\supp X.
\ee
and
\be
\Delta X(F)=0\quad\text{if}\quad\supp F\cap\supp X=\emptyset\ .
\ee
\end{lemma}
\subsection{Perturbative agreement}\label{PA}
Within perturbation theory, the change of the Lagrangian by a symmetry operation of the configuration space can also be treated by the principle of perturbative agreement \cite{HW04} (see also \cite{Zahn15,DHP17}).
This principle amounts to the invariance of the theory under different decompositions of the Lagrangian into a free and an interacting part. It is not clear how this principle can be incorporated into the C*-algebraic construction. In a perturbative construction it can be formulated as follows.

Let $L_0,L$ be 2nd order Lagrangian densities which differ only within some compact region and have normally hyperbolic Euler-Lagrange derivatives $K_0\Phi-q_0$ and $K\Phi-q$, with metrics $g_0,g$ for which $M$ is globally hyperbolic and densities (\textit{sources}) $q_0,q$, respectively. Let $V=\int (L-L_0)$ be the interaction. Local functionals of this form are called admissible interactions for the Lagrangian $L_0$, likewise $(-V)$ is an admissible interaction for $L$. Let $\Omega:\Ec(M,\RR^n)\to\Ec(M,\RR^n)$ denote the retarded
M\o ller map \cite{DHP17},
\begin{equation}
    K_0\Omega(\phi)-q_0= K\phi-q\ ,\  \supp(\Omega(\phi)-\phi)\subset J_+(\supp V)\ ,
\end{equation}
\ie $\Omega(\phi)=\phi+\Delta^R_0((K-K_0)\phi+q_0-q)$ with the retarded Green operator $\Delta^R_0$ for $K_0$.
We choose $\star$-products $\star_0$, $\star$ with Hadamard functions $H_0$ and $H$ such that 
\begin{equation}
   \Omega' H(\Omega')^t=H_0\ . 
\end{equation}
with the derivative $\Omega'=1+\Delta^R_0(K-K_0)$ of $\Omega$. Then the two star-products are related by
\begin{equation}\label{eq:starintertwining}
    F\star_0G[\Omega(\phi)]=(F\circ\Omega\star G\circ\Omega)[\phi]\ .
\end{equation}
Let now $S_0$ and $S$ be perturbative S-matrices which satisfy the causality condition and the dynamical condition for the respective Lagrangians and $\star$-products. We then use $S_0$ to define a second 
S-matrix $\hat{S}$ for $L$ by 
\be\label{eq:S_V}
\hat{S}(F)[\phi]=\bigl(S_0(V)^{-1}\star_0S_0(V+F)\bigr)[\Omega\phi]\doteq S_V(F)[\Omega\phi]\ .
\ee
The (given) pair $(S,S_0)$ satisfies perturbative agreement iff $S=\hat S$.
For linear functionals $F$ the perturbative S-matrices $\hat{S}$ and $S$ coincide as a consequence of the dynamical relation for $S$ (see \eg \cite{BF21}). For more general functionals we find
\begin{lemma}
    $\hat{S}$ satisfies the conditions for $L$ with $\star$-product $\star$.
\end{lemma}
\begin{proof}
    We start with the dynamical condition: since
    $V+\delta L(\psi)=V^\psi+\delta L_0(\psi)$, the dynamical relation for $S_0$ reads:
    \be
    S_0\bigl(V+F^\psi+\delta L(\psi)\bigr)=S_0(V+F)\star_0 S_0\bigl(\delta L_0(\psi)\bigr)=S_0\bigl(\delta L_0(\psi)\bigr)\star_0 S_0(V+F)\ .
    \ee
    By using this relation and, in a later step, the same relation with $F=0$, we obtain
    \begin{equation}
    \begin{split}
        \hat{S}\bigl(F^\psi+\delta L(\psi)\bigr)[\phi]&=\Bigl(S_0(V)^{-1}\star_0S_0\bigl(V+F^\psi+\delta L(\psi)\bigr)\Bigr)[\Omega\phi]\\
        &=\Bigl(S_0(V)^{-1}\star_0 S_0\bigl(\delta L_0(\psi)\bigr)\star_0 S_0(V)\star_0 S_0(V)^{-1}\star_0 S_0(V+F)\Bigr)[\Omega\phi]\\
        &=\Bigl(S_0(V)^{-1}\star_0 S_0\bigl(V+\delta L(\psi)\bigr)\star_0 S_0(V)^{-1}\star_0 S_0(V+F)\Bigr)[\Omega\phi]\\
        &=\Bigl(S_V\bigl(\delta L(\psi)\bigr)\star_0 S_V(F)\Bigr)[\Omega\phi]\\
        &=\Bigl(\hat{S}\bigl(\delta L(\psi)\bigr)\star \hat{S}(F)\Bigr)[\phi]\ .
    \end{split}
    \end{equation}
Proceeding analogously one verifies also that 
$\hat{S}\bigl(F^\psi+\delta L(\psi)\bigr)=\hat{S}(F)\star \hat{S}\bigl(\delta L(\psi)\bigr)$.

    For the causality condition  it suffices to check the 2-factor relation. Let $\supp F\cap J_-(\supp G)=\0$. We have
    \begin{equation}
        \begin{split}
            \hat{S}(F+G)[\phi]&=(S_0(V)^{-1}\star_0S_0(V+F+G))[\Omega\phi]\\
            &=(S_0(V)^{-1}\star_0S_0(V+F)\star_0S_0(V)^{-1}\star_0S_0(V+G))[\Omega\phi]\\
            &=(S_V(F)\star_0S_V(G))[\Omega\phi]\\
            &=(\hat{S}(F)\star \hat{S}(G))[\phi]
        \end{split}
    \end{equation}
    where we used \eqref{eq:starintertwining} in the last step.
\end{proof}
By the Main Theorem \eqref{eq:MainT}-\eqref{eq:Z-perturbative} we conclude that again the two S-matrices $S$ and $\hat{S}$ are related by a 
renormalization group transformation $\zeta$,
but now the condition of compact support has to be weakened to past compact support.

To indicate the dependence on $V$ and on the action $[L]$ corresponding to the Lagrangian $L$ we introduce the notations
\be
S=S_{[L]},\,\, S_0=S_{[L_0]}=S_{[L]-V}\,\text{ and }\,\hat{S}=S_{[L]}^V
\ee
as well as
\be
K_V=K_0,\,\, q_V=q_0\, \text{ and }\,\Omega=\Omega_V
 \ee
 and we shall use the factorization
 \be\label{eq:Omega-factorization}
 \Omega_{V+W}=\Omega_W^V\circ\Omega_V \ \text{ with }K_{V+W}\,\Omega_{W}^V(\phi)-q_{V+W}=K_V\phi-q_V\ 
 \ee
where $(-W)$ is an admissible interaction for $[L]-V$ (hence $-(V+W)$ is admissible for $[L]$, cf. \cite[Sect.~2.1]{BDFR21}). 
We define the renormalization group elements $\zeta^{[L]}_V$ and $\zeta^{[L]-V}_W$ by
\be\label{eq:obstructionPA}
S_{[L]}^V=S_{[L]}\zeta^{[L]}_V\ ,\ S_{[L]-V}^W=S_{[L]-V}\zeta^{[L]-V}_W\ ;
\ee
note in particular that $\zeta^{[L]}_V(0)=0$ (and similarly for $\zeta^{[L]-V}_W$).
\begin{proposition}
The renormalization group elements \eqref{eq:obstructionPA} which characterize the obstruction to perturbative agreement satisfy the cocycle relation
\be\label{eq:cocycle-PA}
\zeta^{[L]}_{V+W}(F)=\zeta^{[L]}_V\bigl(\zeta^{[L]-V}_W(V+F)-\zeta^{[L]-V}_W(V)\bigr)\ .
\ee
\end{proposition}

\begin{remark}
Writing $\Rc_0([L])\doteq\Rc_0$ to make visible the dependence of the St\"uckelberg-Petermann RG on the underlying free Lagrangain $L$,
we point out that the cocycle relation and $\zeta^{[L]}_V,\,\zeta^{[L]}_{V+W}\in\Rc_0([L])$ imply that 
$\bigl(\zeta^{[L]-V}_W(V+\bullet)-\zeta^{[L]-V}_W(V)\bigr)\in\Rc_0([L])$ (in contrast to $\zeta^{[L]-V}_W\in\Rc_0([L]-V)$).
\end{remark}

\begin{proof}
By definition (in particular \eqref{eq:S_V}) we obtain
\begin{equation}
\begin{split}
&S_{[L]}\zeta^{[L]}_{V+W}(F)[\phi]=S_{[L]}^{V+W}(F)[\phi]\\
&=\bigl(S_{[L]-V-W}(V+W)^{-1}\star S_{[L]-V-W}(V+W+F)\bigr)[\Omega_{V+W}\phi]\\
&=\bigl(S_{[L]-V-W}(V+W)^{-1}\star S_{[L]-V-W}(W)\star S_{[L]-V-W}(W)^{-1}\star S_{[L]-V-W}(V+W+F)\bigr)[\Omega_{V+W}\phi]\label{eq:cocycleS}
\end{split}
\end{equation}
We use the splitting $\Omega_{V+W}=\Omega_W^V\circ\Omega_V$ and get
\begin{align}
\bigl(S_{[L]-V-W}(W)^{-1}\star S_{[L]-V-W}(V+W+F)\bigr)[\Omega_{V+W}\phi]=\,&\,S_{[L]-V}^W(V+F)[\Omega_V\phi]\nonumber\\
=\,&\,S_{[L]-V}\zeta_W^{[L]-V}(V+F)[\Omega_V\phi]\ .
\end{align}
as well as the corresponding expression for $F=0$. (Note that in the formulas written so far, ``$\star$'' belongs to the free theory given by
$[L]-V-W$ and the upper index $(-1)$ denotes the inverse w.r.t.~that star product.)
We then take into account that for the 2nd order functional $V$ renormalization group elements $Z$ just add a constant, $Z(V)=V+c$, hence
\be
S_{[L]-V}\zeta_W^{[L]-V}(V)=S_{[L]-V}(V+c)\ 
\ee 
and 
\be
S_{[L]-V}\zeta_W^{[L]-V}(V+F)=S_{[L]-V}\bigl(\zeta_W^{[L]-V}(V+F)-\zeta_W^{[L]-V}(V)+V+c\bigr)\ .
\ee 
By inserting these relations into \eqref{eq:cocycleS} and using  \eqref{eq:starintertwining}, \eqref{eq:S_V}, 
\eqref{eq:obstructionPA} and $S_{[L]-V}(G+c)=S_{[L]-V}(G)\,e^{ic}$ (for all $G\in\Floc(M)$), we obtain the claim.
\end{proof}

The cocycle characterizing the obstruction to perturbative agreement is trivial, in agreement with the arguments in \cite{HW04,DHP17}:
\begin{proposition}\label{prop:PA-anomaly-free} 
Let $\tilde{S}_{[L]-V}\doteq S_{[L]-V}Z_V$ with $Z_V(F)\doteq (\zeta_V^{[L]})^{-1}(F-V)-(\zeta_V^{[L]})^{-1}(-V)$ for 
all admissible 2nd order interactions $-V$ of the action $[L]$. Then $\tilde{S}_{[L]-V}$ satisfies the condition of perturbative agreement, 
\ie for all admissible 2nd order perturbations $-W$ of the action $[L]-V$ we have $\tilde{S}_{[L]-V}^W=\tilde{S}_{[L]-V}$. 
\end{proposition}
\begin{proof}
By using \eqref{eq:S_V}, \eqref{eq:obstructionPA} and $S_{[L]-V-W}(G+c)=S_{[L]-V-W}(G)\,e^{ic}$, we get
\begin{align}
   \tilde{S}_{[L]-V}^W&(V+F)[\phi]=\big(\tilde{S}_{[L]-V-W}(W)^{-1}\star\tilde{S}_{[L]-V-W}(V+W+F)\big)[\Omega_W^V\phi]\\
   &=\bigl(S_{[L]-V-W}(Z_{V+W}(W))^{-1}\star S_{[L]-V-W}(Z_{V+W}(V+W+F))\bigr)[\Omega_W^V\phi]\\
   &=\bigl(S_{[L]-V-W}(W)^{-1}\star S_{[L]-V-W}(Z_{V+W}(V+W+F)-Z_{V+W}(W)+W)\bigr)[\Omega_W^V\phi]\\
   &=\bigl(S_{[L]-V}^W(Z_{V+W}(V+W+F)-Z_{V+W}(W))\bigr)[\phi]\\
   &=S_{[L]-V}\zeta_W^{[L]-V}\bigl(Z_{V+W}(V+W+F)-Z_{V+W}(W)\bigr)[\phi]\\
   &=S_{[L]-V}\zeta_W^{[L]-V}\bigl((\zeta_{V+W}^{[L]})^{-1}(F)-(\zeta_{V+W}^{[L]})^{-1}(-V)\bigr)[\phi]\ .
 \end{align}
 The expression in the last line is equal to $S_{[L]-V}Z_V(V+F)=\tilde{S}_{[L]-V}(V+F)$; this follows from 
 the cocycle relation, as shown in the following lemma.
 \end{proof}
 \begin{lemma}
 The cocycle relation \eqref{eq:cocycle-PA} implies the identity
  \begin{equation}
      (\zeta_{V+W}^{[L]})^{-1}(F)-(\zeta_{V+W}^{[L]})^{-1}(-V)=(\zeta_W^{[L]-V})^{-1}\bigl((\zeta_V^{[L]})^{-1}(F)-(\zeta_V^{[L]})^{-1}(-V)\bigr)\ .
  \end{equation}   
 \end{lemma}
 \begin{proof}
    The cocycle relation implies that the inverse of $\zeta_{V+W}^{[L]}$
    is obtained by solving the equation 
    \begin{equation}
        F=\zeta^{[L]}_V\bigl(\zeta^{[L]-V}_W(V+G)-\zeta^{[L]-V}_W(V)\bigr)
    \end{equation}
    for $G$, hence
    \begin{equation}
         (\zeta_{V+W}^{[L]})^{-1}(F)+V=(\zeta_W^{[L]-V})^{-1}\bigl((\zeta_V^{[L]})^{-1}(F)+\zeta_W^{[L]-V}(V)\bigr)\ .
     \end{equation}
     For $F=-V$ we get an identity for c-numbers,
    \begin{equation}
         (\zeta_{V+W}^{[L]})^{-1}(-V)+V=(\zeta_W^{[L]-V})^{-1}\bigl((\zeta_V^{[L]})^{-1}(-V)+\zeta_W^{[L]-V}(V)\bigr)\ .
     \end{equation} 
     We subtract this equation from the equation for generic $F$, use again the fact that renormalization group elements commute with the addition of c-numbers and obtain the equation in the lemma,
     \begin{equation}
     \begin{split}
         &(\zeta_{V+W}^{[L]})^{-1}(F)-(\zeta_{V+W}^{[L]})^{-1}(-V)\\
         &\qquad =(\zeta_W^{[L]-V})^{-1}\bigl((\zeta_V^{[L]})^{-1}(F)+\zeta_W^{[L]-V}(V)-(\zeta_V^{[L]})^{-1}(-V)-\zeta_W^{[L]-V}(V)\bigr)\\
         &\qquad =(\zeta_W^{[L]-V})^{-1}\bigl((\zeta_V^{[L]})^{-1}(F)-(\zeta_V^{[L]})^{-1}(-V)\bigr)\ .
     \end{split}    
     \end{equation}
\end{proof}

We specialize now to the case that the interaction is of the form $V=\delta_{\g}L$ for $\g\in\Gc_c(M)$. In this case we have (see \cite{Zahn15}, Prop. 2.10)
\begin{proposition}
   $\Omega_V(\phi)=\g(\phi)$ for on shell configurations $\phi$ (with respect to the Lagrangian $\g_{\ast}L$). 
\end{proposition}
\begin{proof}
By construction of $\Omega_V$, $\Omega_V(\phi)$ satisfies the field equation for $L$, $K\Omega_V(\phi)=q$. By assumption, $\phi$ satisfies the field equation for $\g_{\ast}L$, $\g^t K\g(\phi)=\g^t q$. But then $\g(\phi)$ satisfies also the field equation for $L$. Since $\g(\phi)$ coincides with $\Omega_V(\phi)$ outside of $J_+(\supp \g)$, both configurations coincide everywhere.   
\end{proof}
For the inverse M\o ller map we find $\Omega_V^{-1}(\phi)=\g^{-1}(\phi)\doteq\g_{\ast}\phi$, for on shell configurations with respect to $L$.
We observe that \eqref{eq:starintertwining} yields a further isomorphism $\gamma_\g$ between the theories for $L$ and $\g_{\ast}L$, 
$\gamma_g\bigl(S_{[L]}(F)\bigr)=S_{[L]}(F)\circ\Omega_V$,%
\footnote{In particular note that causal factorization of $S_{[L]}(F)$ with respect to $\star_{[L]}$ implies causal factorization
of $\gamma_g\bigl(S_{[L]}(F)\bigr)$ with respect to $\star_{[\g_\ast L]}$, due to \eqref{eq:starintertwining}.}
in addition to the isomorphism 
$$\alpha_\g:S_{[L]}(F)\mapsto S_{[\g_{\ast}L]}(\g_{\ast}F)$$ 
and the inverse isomorphism
$$\beta_{\g}:S_{[\g_{\ast}L]}(F)\mapsto S_{[L]}(\delta_\g L)^{-1}\star S_{[L]}(\delta_\g L+F)\ .$$ 
Perturbative agreement means $\gamma_\g\beta_\g=\mathrm{id}$ on the theory given by $\g_\ast L$, as we see by comparing
\begin{align*}
\gamma_\g\beta_\g\bigl(S_{[L]}(F)\bigr)=\bigl(S_{[L]}(\delta_\g L)^{-1}\star S_{[L]}(\delta_\g L+F)\bigr)\circ\Omega_V
\end{align*}
with \eqref{eq:S_V};
whereas the UAMWI deals with the automorphism $\beta_{\g}\alpha_{\g}$ on the theory given by $L$ (see \cite[Sect.~5]{BDFR21}). 
If perturbative agreement holds, we have $\beta_{\g}\alpha_{\g}=\gamma_{\g^{-1}}\alpha_{\g}$, thus we get from the on-shell UAMWI (by using 
$\zeta_{\g}(0)=0$) the identity
\begin{align}
    S_{[L]}(\zeta_\g(F))[\phi]=\,&\,\bigl(S_{[L]}(\zeta_{\g}(0))^{-1}\star S_{[L]}(g_\ast F+\delta_\g L)\bigr)[\phi]\\
    =\,&\,\bigl(S_{[L]}(\delta_\g L)^{-1}\star S_{[L]}(g_\ast F+\delta_\g L)\bigr)[\phi]\\
    =\,&\,S_{[\g_\ast L]}(\g_\ast F)[\g_\ast \phi]
\end{align}
for on shell configurations $\phi$. 
Hence, the anomaly refers directly to a possible discrepancy between the action of the symmetry on the time ordered products and on the Lagrangian. As a simple example one may look at the scaling transformation for the massless free theory. There the Lagrangian is invariant, but the renormalized time ordered products are not invariant, in general.
In cases where such a discrepancy can be excluded, i.e., it holds that 
\begin{equation}\label{eq:covariance}
S_{[\g_\ast L]}(\g_\ast F)[\g_\ast \phi]=S_{[L]}(F)[\phi]\ ,
\end{equation}
perturbative agreement implies the absence of anomalies, and one can derive the existence of conserved currents. This was used \eg in \cite{HW04} to prove the covariant conservation law for the energy momentum tensor on the basis of a generally covariant renormalization prescription for time-ordered products, and analogously, in \cite{Zahn15} for the covariant conservation law of currents
for gauge symmetries. In these references there is an obstruction to fulfil perturbative agreement. Actually, as our analysis shows, in the presence of nontrivial anomalies, perturbative agreement and covariance in the sense of \eqref{eq:covariance} cannot both be satisfied. This also motivates not to postulate covariance a priori.

\section{Review of the Wess-Zumino consistency relations}\label{sec:WZ}
Following \cite{AG85}, we concentrate here on the subgroup $\G_o\subset\G_c$ of orthogonal field redefinitions
$\phi\mapsto \g\phi$ where $\g:M\mapsto \mathrm{SO}(n,\RR)$
is smooth and compactly supported and $\phi$ is written as a column vector. An orthogonal field redefinition may be interpreted as a gauge transformation. It transforms the trivial connection on the vector bundle $M\times\RR^n$ to an equivalent one which may be considered as an external gauge field $A$ which is a pure gauge, \ie $A=\g^{-1}d\g$ for some $\g\in\G_o$. We consider Lagrangians $L_A$,
\begin{equation}
    L_A(\phi)=\frac12 \langle(d+A)\phi,(d+A)\phi\rangle\,,
\end{equation}
which depend on a compactly supported external gauge field $A$, considered as a $\mathfrak{so}(n,\RR)$-valued 1-form, in symbols $A\in\Omega_c^1(M,\mathfrak{so}(n))$. The bracket here combines the spacetime metric on 1-forms together with the canonical inner product on $\RR^n$. 
We have 
\be\label{eq:LAg}
(\g_{\ast}L_A)(\phi)=L_A(\g^{-1}\phi)=L_{A^{\g}}(\phi)\,,
\ee
where $A^{\g}=\g (d\g^{-1})+\g A\g^{-1}$ is the gauge transformed gauge field. 
Let $V(A)\doteq\int(L_A-L)$ denote the interaction induced by $A$.
Then 
\be\label{eq:V(Ag)=gLV(A)}
V(A^{\g})=\g_L V(A)\,,
\ee
where $\g_L$ refers to the action of $\G_o$ on $\phi$ defined in \eqref{eq:g_L}. One now considers the effective action,
\ie the Legendre transform of the generating functional of connected Green's functions,  
\be
\Gamma(A,\varphi)\doteq\langle\varphi,J\rangle-i\log\bigl(S(V(A)+\langle\Phi,J\rangle)[\phi=0]\bigr)\ ,\ \varphi\in\Dc(M,\RR^n)\ ,
\ee
where $J(\varphi,A)$ is the solution of 
\be
\varphi=-i\frac{\delta}{\delta j}\log\bigl(S(V(A)+\langle\Phi,j\rangle)[\phi=0]\bigr)\big\vert_{j=J}\ .
\ee
Since $L_A$ is a quadratic functional of $\phi$, we have the explicit solution
\be
J=\square_A\varphi
\ee
with the d'Alembertian $\square_A$ for an external gauge field $A$.
Thus 
\be\label{eq:Gamma(A,varphi)}
\Gamma(A,\varphi)=\int L_A(\varphi)-i\log\bigl(S(V(A)+\langle\Phi,\square_A\varphi\rangle)[\phi=0]\bigr)\ .
\ee

In the absence of anomalies $\Gamma$ should be gauge invariant. The action of the gauge group on gauge fields $A\mapsto A^\g$ and matter fields $\varphi\mapsto\g\varphi$ induces a corresponding representation $X\mapsto\partial_X^{A,\varphi}$ of the Lie algebra, acting by derivations on functions 
$K$ of these fields,
\begin{equation}\label{eq:dXA}
    \partial_X^{A,\varphi}K(A,\varphi)\doteq\frac{d}{d\lambda}\big\vert_{\lambda=0}K(A^{\g_{\lambda}},\g_{\lambda}\varphi) 
\end{equation}
with  $\g_{\lambda}=\exp{(-\lambda X)}$; in particular $[\partial_X^{A,\varphi},\partial_Y^{A,\varphi}]=\partial_{[X,Y]}^{A,\varphi}$ holds.
One defines the anomaly by
\be\label{eq:def-anomaly}
G(X,A)\doteq\partial_X^{A,\varphi}\Gamma(A,\varphi)\ .
\ee
Even though $\Gamma$ is nonlocal and depends on $\varphi$, the anomaly is a local functional of $A$ and independent of $\varphi$. These 
two statements are well known from the literature (see e.g.~\cite{Stora87}), but for completeness we give independent proofs below.
An immediate consequence of the definition \eqref{eq:def-anomaly} is the \emph{Wess-Zumino consistency relation}
\be\label{eq:Cc-G(X)}
\partial_X^{A}G(Y,A)-\partial_Y^{A}G(X,A)=G([X,Y],A)\ ,
\ee
where we write $\partial_X^{A}$ instead of $\partial_X^{A,\varphi}$ when acting on a functional not depending on $\varphi$.
The consistency relation is a nontrivial restriction on the structure of anomalies, although it is an obvious consequence of the fact that anomalies are defined directly through a Lie algebra action.

\medskip

Next we show how this consistency relation can be derived directly from the UAMWI.
For $\g\in\G_o$, we have that
\be
V(A^{\g})+\langle\Phi,\square_{A^{\g}}(\g\varphi)\rangle=\g_L\bigl(V(A)+\langle\Phi,\square_A\varphi\rangle\bigr)\ .
\ee
Using the UAMWI, we get
\be\label{eq:UAMWI-VA}
S\bigl(V(A^{\g})+\langle\Phi,\square_{A^{\g}}(\g\varphi)\rangle\bigr)\big|_{\phi=0}=
S\bigl(\zeta_{\g}\bigl(V(A)+\langle\Phi,\square_A\varphi\rangle\bigr)\bigr)\big|_{\phi=0}\ .
\ee 
Using parts $(i)$ and $(iii)$ of Prop.~4.14 in \cite{BDFR21}, we obtain 
\be\label{eq:WZa}
\zeta_{\g}\bigl(V(A)+\langle\Phi,\square_A\varphi\rangle\bigr)=
\zeta_{\g}\bigl(V(A)\bigr)+\langle\Phi,\square_A\varphi\rangle=V(A)+\langle\Phi,\square_A\varphi\rangle+\fG(\g,A)\,,
\ee
with a functional $\fG$ which depends on $\g$ and $A$ but does neither depend on the field $\varphi$ nor on the field configuration $\phi$. 
Moreover $\fG$ is local in the sense that for $x\not= y$
\be\label{eq:G-local}
\frac{\delta^2\fG}{\delta A(x)\delta A(y)}=0\quad,\quad\frac{\delta^2\fG}{\delta \g(x)\delta A(y)}=0 \quad\text{and} 
\quad\frac{\delta^2\fG}{\delta \g(x)\delta \g(y)}=0\ .
\ee

This follows from $\fG(\g,A)=\zeta_{\g}(V(A))-V(A)$ \eqref{eq:WZa} and the following proposition:
\begin{proposition}
Let $\omega,\omega_1,\omega_2\in\Omega_c^1(M,\mathfrak{so}(n))$, $\g,\g_1,\g_2,\h\in\G_0$ such that
$\supp\omega_1\cap\supp\omega_2=\0$, $\supp\omega\cap\supp\g=\0$ and $\supp\g_1\cap\supp\g_2=\0$.Then
\begin{enumerate}
    \item $\zeta_\h\bigl(V(A+\omega_1+\omega_2)\bigr)=\zeta_\h\bigl(V(A+\omega_1)\bigr)-\zeta_\h\bigl(V(A)\bigr)+\zeta_\h\bigl(V(A+\omega_2)\bigr)$\ ,
    \item $\zeta_{\g\h}\bigl(V(A+\omega)\bigr)=\zeta_\h\bigl(V(A+\omega)\bigr)-\zeta_\h\bigl(V(A)\bigr)+\zeta_{\g\h}\bigl(V(A)\bigr)$\ ,
    \item $\zeta_{\g_1\g_2\h}\bigl(V(A)\bigr)=\zeta_{\g_1\h}\bigl(V(A)\bigr)-\zeta_\h\bigl(V(A)\bigr)+\zeta_{\g_2\h}\bigl(V(A)\bigr)$\ .
\end{enumerate}
\end{proposition}
\begin{proof} 
\begin{enumerate}
    \item From $V(A)=\int(L_A-L)$ we see that $V$ is a local functional of $A$, hence
    \begin{equation}
        V(A+\omega_1+\omega_2)=V(A+\omega_1)-V(A)+V(A+\omega_2)
    \end{equation}
    and $\supp(V(A+\omega_i)-V(A))\subset\supp\omega_i$, $i=1,2$. 
    Since $\zeta_\h$ satisfies the additivity relation
    \begin{equation}\label{eq:additivity}
        \zeta_\h(F+G+H)=\zeta_\h(F+G)-\zeta_\h(G)+\zeta_\h(G+H)
    \end{equation}
    for $F,G,H\in\Floc(M)$ with $\supp F\cap\supp H=\0$, we get
    \begin{equation}
        \begin{split}
            \zeta_\h\bigl(V(A+\omega_1+\omega_2)\bigr)&=\zeta_\h\bigl((V(A+\omega_1)-V(A))+V(A)+(V(A+\omega_2)-V(A))\bigr)\\
            &=\zeta_\h\bigl(V(A+\omega_1)\bigr)-\zeta_{\h}\bigl(V(A)\bigr)+\zeta_\h\bigl(V(A+\omega_2)\bigr)\ .
        \end{split}
    \end{equation}
    
    \item From the cocycle relation we have $\zeta_{\g\h}=\zeta_\h\zeta_\g^\h$ with $\zeta_\g^\h\doteq\h_L^{-1}\zeta_\g\h_L\in\Rc_c$
    \cite[Lemma 5.4]{BDFR21}. In the first step we show that $\supp\zeta_g^\h\subset\supp\zeta_\g$. 
    
    Let $F,G\in\Floc(M)$ and $\supp G\cap\supp\g=\0$. With $\supp \h_{\ast}G=\supp G$ we get
    \begin{equation}
    \begin{split}
    \zeta_\g^\h(F+G)&=\h_L^{-1}\zeta_\g(\h_LF+\h_{\ast}G)\\
    &=\h_L^{-1}\bigl(\zeta_\g(\h_LF)+\h_{\ast}G\bigr)\\
    &=\zeta_\g^\h(F)+G\ .
    \end{split}
    \end{equation}

    Then $\supp\bigl(\zeta_\g^\h(V(A))-V(A)\bigr)\subset\supp\zeta_\g^\h\subset\supp\g$ (see \cite[Prop.~4.14(ii)]{BDFR21}) and we find that
    \begin{equation}
        \begin{split}
            \zeta_{\g\h}\bigl(V(A+\omega)\bigr)&=\zeta_\h\zeta_\g^\h\bigl((V(A+\omega)-V(A))+V(A)\bigr)\\
            &=\zeta_\h\bigl(\zeta_\g^\h(V(A))+V(A+\omega)-V(A)\bigr)\\
            &=\zeta_\h\bigl(\bigl(\zeta_\g^\h(V(A))-V(A)\bigr)+V(A)+\bigl(V(A+\omega)-V(A)\bigr)\bigr)\\
            &=\zeta_\h\zeta_\g^\h\bigl(V(A)\bigr)-\zeta_\h\bigl(V(A)\bigr)+\zeta_\h\bigl(V(A+\omega)\bigr)\ .
        \end{split}
    \end{equation}
    \item In the first step we show that for $\g,\h\in\G_o$ with disjoint supports (i.e. $\supp\g\cap\supp\h=\0$) the following relation holds for $F\in\Floc(M)$:
    \begin{equation}
        \begin{split}
            \zeta_\g^\h(F)&=\h_L^{-1}\zeta_\g\h_L(F)\\&=\h_L^{-1}\zeta_\g\bigl((\h_LF-F)+F\bigr)\\
            &=\h_L^{-1}(\h_LF-F+\zeta_\g(F))\\
            &=F-\h_L^{-1}F+\h_L^{-1}\zeta_\g(F)\\
            &=F-\h_L^{-1}F+\h_L^{-1}\bigl((\zeta_\g(F)-F)+F\bigr)\\
            &=F-\h_L^{-1}F+\zeta_\g(F)-F+\h_L^{-1}F\\
            &=\zeta_\g(F)\,.
        \end{split}
    \end{equation}
    Here we used that $\supp(\h_LF-F)\subset\supp\h$ and $\supp(\zeta_\g(F)-F)\subset\supp\g$.
    We now compute
    \begin{equation}
        \begin{split}
            \zeta_{\g_1\g_2\h}(F)&=\zeta_\h\h_L^{-1}\zeta_{\g_2}\zeta_{\g_1}^{\g_2}\h_L(F)\\
            &=\zeta_\h\h_L^{-1}\zeta_{\g_2}\zeta_{\g_1}\h_L(F)\\
            &=\zeta_\h\h_L^{-1}\zeta_{\g_2}\bigl((\zeta_{\g_1}\h_L(F)-\h_L(F))+\h_L(F)\bigr)\\
            &=\zeta_\h\h_L^{-1}\bigl((\zeta_{\g_1}\h_L(F)-\h_L(F))+\zeta_{\g_2}\h_L(F)\bigr)\\
            &=\zeta_\h\bigl(\zeta_{\g_1}^\h(F)-F+\zeta_{\g_2}^\h(F)\bigr)\\
            &=\zeta_\h\bigl((\zeta_{\g_1}^\h(F)-F)+F+(\zeta_{\g_2}^\h(F)-F)\bigr)\\
            &=\zeta_\h\zeta_{\g_1}^\h(F)-\zeta_\h(F)+\zeta_\h\zeta_{\g_2}^\h(F)\\
            &=\zeta_{\g_1\h}(F)-\zeta_\h(F)+\zeta_{\g_2\h}(F)\ .
        \end{split}
    \end{equation}
    In the second last step we were able to apply the additivity of $\zeta_\h$ \eqref{eq:additivity} 
    since $\supp(\zeta_{\g_j}^\h(F)-F)\subset\supp\g_j$.
    Inserting $F=V(A)$ yields the claim in the proposition.
\end{enumerate}
\end{proof}

We now insert \eqref{eq:WZa} into \eqref{eq:UAMWI-VA} and find
\be
S\bigl(V(A^{\g})+\langle\Phi,\square_{A^{\g}}(\g\varphi)\rangle\bigr)\big|_{\phi=0}=
S\bigl(V(A)+\langle\Phi,\square_A\varphi\rangle\bigr)\big|_{\phi=0}\,e^{i\fG(\g,A)}.
\ee
Hence, using \eqref{eq:Gamma(A,varphi)} and \eqref{eq:LAg}, we obtain the following action of $\G_o$ on the effective action:
\be
\Gamma(A^{\g},\g\varphi)=\Gamma(A,\varphi)+\fG(\g,A)\,.
\ee
Since $(A^\h)^\g = A^{\g\h}$ and $\fG(\g,A^{h})$ does not depend on $\varphi$,
we immediately see that $\fG$ satisfies the cocycle relation
\be\label{eq:Cc-G(g)}
\fG(\g\h,A)=\fG(\g,A^{h})+\fG(\h,A)\ .
\ee
The Wess-Zumino consistency relation is just the infinitesimal version of this cocycle relation with
\begin{equation}
    G(X,A)=\frac{d}{d\lambda}\Big\vert_{\lambda=0}\fG(\exp(-\lambda X),A)\ .
\end{equation}
This follows from the fact that $\G_o$ acts on $\Gamma$ as an antirepresentation.
To see it directly, we compute
\begin{equation}
    \begin{split}
        \partial_X^AG(Y,A)-\partial_Y^AG(X,A)&=\frac{\partial^2}{\partial\lambda\partial\mu}\Big\vert_{\lambda=\mu=0}\Bigr(\fG\bigl(\exp(-\mu Y),A^{\exp(-\lambda X)}\bigr)-\fG\bigl(\exp(-\lambda X),A^{\exp(-\mu Y)}\bigr)\Bigl)\\
        =\frac{\partial^2}{\partial\lambda\partial\mu}\Big\vert_{\lambda=\mu=0}&\Bigr(\fG\bigl(\exp(-\mu Y)\exp(-\lambda X),A\bigr)-\fG\bigl(\exp(-\lambda X )\exp(-\mu Y),A\bigr)\Bigl)\ ,
    \end{split}
\end{equation}
where we used the cocycle condition \eqref{eq:Cc-G(g)} and the fact that the terms which depend only on one of the variables, $\lambda$ or
$\mu$, do not contribute to the derivative. We use now the following consequences of the Baker-Campbell-Hausdorff formula:
\begin{equation}
    \exp(-\lambda X)\exp(-\mu Y)=\g_+(\lambda,\mu)\exp(\frac12\lambda\mu[X,Y])
\end{equation}
and
\begin{equation}
    \exp(-\mu Y)\exp(-\lambda X)=\g_-(\lambda,\mu)\exp(-\frac12\lambda\mu[X,Y])\ ,
\end{equation}
where $\g_+$ and $\g_-$ coincide up to 2nd order. Inserting this into the previous formula and using again the cocycle condition, together with $\mathcal{G}(\e,A)\equiv0$, we find
\begin{equation}
    \begin{split}
        \partial_X^AG(Y,A)-\partial_Y^AG(X,A)&=\frac{\partial^2}{\partial\lambda\partial\mu}\Big\vert_{\lambda=\mu=0}
        \Bigr(\fG\bigl(\g_-(\lambda,\mu),A^{\exp(-\frac12\lambda\mu[ X,Y])}\bigr)+\fG\bigl(\exp(-\tfrac12\lambda\mu[X,Y]),A\bigr)\\
        &\phantom{\frac{\partial^2}{\partial\lambda\partial\mu}\Big\vert_{\lambda=\mu=0}}
        -\fG\bigl(\g_+(\lambda,\mu),A^{\exp(+\frac12\lambda\mu[ X,Y])}\bigr)-\fG\bigl(\exp(\tfrac12\lambda\mu[X,Y]),A\bigr)\Bigr)\\
        &=G([X,Y],A)\ ,
    \end{split}
\end{equation}
since the terms involving $\g_+$ and $\g_-$ cancel.

\section{Consistency relation for the anomaly of the AMWI}\label{sec:infinitesimal-cc}
We return to the more general framework introduced in Sect.~\ref{sec:properties}. We aim at a derivation of a consistency relation for the anomaly
map $X\mapsto\Delta X$ \eqref{eq:Delta-X} of the AMWI \eqref{eq:AMWI-q} which holds for \emph{general interactions}
(i.e., any $F\in\Floc(M)$ is admitted), by using only the AMWI.
For this purpose
we consider two paths $\g^{\lambda}$ and $\h^{\mu}$ in $\G_c(M)$ with tangent vectors $X$ and $Y$ at 0, respectively.
We compute
\be
\frac{\partial^2}{i\partial\lambda\partial\mu}\Big\vert_{\lambda=\mu=0}\Bigl(S\bigl((\g^{\lambda}\h^{\mu})_{L_q}F\bigr)-
S\bigl((\h^{\mu}\g^{\lambda})_{L_q}F\bigr)\Bigr)
=S(F)\cdot_T\partial_{[X,Y]}(F+L_q)\ ,
\ee
by using \eqref{eq:bracket}.
According to the AMWI \eqref{eq:AMWI-q}, it coincides for configurations $\phi$ with $\frac{\delta L}{\delta\phi}[\phi]=q$ 
with 
\be
S(F)\cdot_T \Delta[X,Y](F)\ .
\ee
Instead we can also use the AMWI after the first derivative and obtain on those configurations $\phi$
\be
\begin{split}
\frac{\partial}{i\partial\lambda}\Big\vert_{\lambda=0} S\bigl((\g^{\lambda}\h^{\mu})_{L_q}F\bigr)&=S(\h^{\mu}_{L_q}F)\cdot_T\Delta X(\h^{\mu}_{L_q}F)\\
&=\frac{d}{idt}\Big\vert_{t=0} S\bigl(\h^{\mu}_{L_q}(F+t(\h^{\mu}_{\ast})^{-1}\Delta X(\h^{\mu}_{L} F))\bigr)\,,\end{split}
\ee
where we used that $\Delta X$ is invariant under addition of an affine function of the field $\Phi$. (This follows from the defining properties 
$(\mathrm{P}3)$ and $(\mathrm{P}5)$ of $\Delta X\in\LieRc$, cf.~Footnote \ref{fn:DeltaY(F+c)}.)
Taking now the derivative with respect to $\mu$ and using again the AMWI we obtain
\be
\begin{split}
\frac{\partial^2}{i\partial\lambda\partial\mu}\Big\vert_{\lambda=\mu=0} S\bigl((\g^{\lambda}\h^{\mu})_{L_q}F\bigr)
&=\frac{d}{dt}\Big\vert_{t=0} S(F+t\Delta X(F))\cdot_T \Bigl(\Delta Y(F+t\Delta X(F))\\
&\qquad\qquad -t\partial_Y(\Delta X(F)+t\langle(\Delta X)'(F),\partial_Y(F+L)\rangle\Bigr)\\
&=S(F)\cdot_T \Bigl(i\Delta X(F)\cdot_T\Delta Y(F)+\langle\Delta Y'(F),\Delta X(F) \rangle\\
&\quad-\partial_Y(\Delta X(F))+\langle(\Delta X)'(F),\partial_Y(F+L)\rangle\Bigr)
\end{split}
\ee
on the above mentioned configurations.
We finally arrive at a consistency relation which does no longer depend on the source $q$ and therefore holds for all configurations $\phi$.
\begin{theorem}\label{thm:cc}
 The anomaly $\Delta$ of the AMWI satisfies the consistency relation
\begin{equation}
\begin{split}\label{eq:WZ-0}
    \Delta([X,Y])(F)&=\langle(\Delta Y)'(F),\Delta X(F) \rangle-\langle(\Delta X)'(F),\Delta Y(F) \rangle\\ 
    &\quad+\partial_X(\Delta Y(F))-\partial_Y(\Delta X(F))\\
    &\quad-\langle(\Delta Y)'(F),\partial_X(L+F)\rangle +\langle(\Delta X)'(F),\partial_Y(L+F)\rangle
\end{split}
\end{equation}
for $X,Y\in\LieGc$.
\end{theorem}

We call \eqref{eq:WZ-0} the ``extended Wess-Zumino consistency condition,'' because for \emph{quadratic} functionals $F$,
it reduces to the Wess-Zumino condition \eqref{eq:Cc-G(X)}. Namely, for those functionals, $\Delta X(F)$ is a constant functional 
(see part $(iii)$ of \cite[Prop.~4.14]{BDFR21}). Then $\partial_Y\Delta X(F)$
and\footnote{The defining property $(\mathrm{P}5)$ of $\Delta Y\in\LieRc$ implies that $\Delta Y(F+c)=\Delta Y(F)$ for all $F\in\Floc(M),\,c\in\RR$.
\label{fn:DeltaY(F+c)}} 
$\langle\Delta Y'(F),\Delta X(F)\rangle$ vanish for $X,Y\in\LieGc$, and the Wess-Zumino relation is obtained by using the identifications 
\be\label{eq:G=Delta}
G(X,A)=-\Delta X(V(A))
\ee
and
\begin{align}\label{eq:dXG(Y)-DeltaY}
\partial_X^AG(Y,A)\overset{\eqref{eq:G=Delta}}{=}&-\frac{d}{d\lambda}\Big\vert_{\lambda=0}\Delta Y\bigl(V(A^{g^\lambda})\bigr)\nonumber\\
\overset{\eqref{eq:V(Ag)=gLV(A)}}{=}&-\frac{d}{d\lambda}\Big\vert_{\lambda=0}\Delta Y\bigl(\g^\lambda_LV(A)\bigr)\nonumber\\
=\,\,&\langle\Delta Y'(V(A)),\partial_X \bigl( L+V(A)\bigr)\rangle\ ,
\end{align}
where $\g^{\lambda}\doteq\exp{(-\lambda X)}$ as in \eqref{eq:dXA}.

\section{Consistency condition of the AMWI as a cocycle for Lie algebras}
\label{sec:cc-from-cocycle}
The consistency relation for $\Delta$ derived in the previous section actually shows that $\Delta$ is a Lie-algebraic cocycle. This can be most easily seen by starting from the UAMWI with its group theoretical cocycle $\zeta$. 
The cocycle $\zeta$ intertwines two actions of $\G_c(M)$ on $\Floc(M)$,
namely $(\g,F)\mapsto \g_LF$ and $(\g,F)\mapsto \g_L\zeta_{\g}^{-1}(F)$.
They induce representations $R$ and $P$ on the space of functions $K$ on $\Floc(M)$  by
\begin{equation}
    (R(\g)K)(F)=K(\g_L^{-1}F)\ ,\ (P(\g)K)(F)=K(\zeta_{\g}\g_L^{-1}F) 
\end{equation}
The relation $R(\g_1\g_2)=R(\g_1)R(\g_2)$ relies on $(\g_1\g_2)_L=\g_{1,L}\,\g_{2,L}$; to obtain $P(\g_1\g_2)=P(\g_1)P(\g_2)$ we additionally
use our crucial input: the cocycle relation for $\zeta$ \eqref{eq:cocycle}.

The corresponding representations $r$ and $p$ of $\LieGc$ act by derivations on smooth functions on $\Floc(M)$, \ie
\begin{equation}\label{eq:Lie-rep}
    r(X)K(F)=\langle K'(F),-\partial_{XF}-\partial_XL\rangle\ ,\ p(X)K(F)=\langle K'(F),-\partial_XF-\partial_XL+\Delta X(F)\rangle\,,
\end{equation}
where we used \eqref{eq:Delta-X}.
Representations $r$ and $p$ differ by the linear map $X\mapsto q(X)=p(X)-r(X)$ with $q(X)K(F)=\langle K'(F),\Delta X(F)\rangle$.
Since $p$ and $r$ are representations, $q$ satisfies the relation
\begin{equation}\label{eq:relation-q}
    q([X,Y])=[q(X),q(Y)]+[r(X),q(Y)]-[r(Y),q(X)]
\end{equation}
It remains to compute the commutators of these derivations.
We obtain
\begin{equation}\label{eq:qXqY}
\begin{split}
    (q(X)(q(Y)K))(F)&=\frac{d}{d\lambda}\Big\vert_{\lambda=0}\,
    (q(Y)K)(F+\lambda\Delta X(F))\\
    &=\frac{d}{d\lambda}\Big\vert_{\lambda=0}\,
    \langle K'(F+\lambda\Delta X(F)),\Delta Y(F+\lambda\Delta X(F))\rangle\\
    &=\langle K''(F),\Delta X(F)\otimes \Delta Y(F)\rangle+\big\langle K'(F),\langle (\Delta Y)'(F),\Delta X(F)\rangle\big\rangle\ ,
\end{split}
\end{equation}
hence
\begin{equation}\label{eq:[qX,qY]}
\begin{split}
    [q(X),q(Y)]K(F)&=\big\langle K'(F),\langle (\Delta Y)'(F),\Delta X(F)\rangle-\langle (\Delta X)'(F),\Delta Y(F)\rangle\big\rangle\\ 
    &=\langle K'(F), [\Delta Y,\Delta X]_{\LieRc}(F)\rangle\ ,
\end{split}
\end{equation}
where we use the explicit formula for the Lie bracket in $\LieRc$ derived in \cite[App.~C]{BDFR21}. 
Proceeding analogously to 
\eqref{eq:qXqY}-\eqref{eq:[qX,qY]} we get 
\begin{equation}\label{eq:[r,q]}
\begin{split}
    [r(X),q(Y)]K(F)&=\big\langle K'(F),-\langle(\Delta Y)'(F),\partial_XF+\partial_XL\rangle
        +\partial_X(\Delta Y(F))\big\rangle\\
        &=\big\langle K'(F),(\partial_X\Delta Y)(F)\big\rangle
\end{split}
\end{equation}
with the representation $X\mapsto\partial_X$, \be(\partial_X z)(F)\doteq\partial_X(z(F))-\langle z'(F),\partial_X(F+L)\rangle
\ee of $\LieGc$
by derivations on $\LieRc$.%
\footnote{To see that $X\to\partial_X$ is indeed a representation, note that it is the infinitesimal version of the representation $D$ of 
$\G_c(M)$ on the space of maps $K:\Floc(M)\to\Floc(M)$ defined by $D(\g)K(F)=\g_\ast K(\g_L^{-1}F)$.}
We also have the analogous relation with $X$ and $Y$ interchanged, so combining the two, we arrive at precisely the same consistency condition for the anomaly $\Delta$ 
of the AMWI as in Theorem \ref{thm:cc} which now assumes the form:

\begin{theorem}\label{thm:Liecc}
The cocycle relation \eqref{eq:cocycle} for the anomaly $\zeta$ of the UAMWI implies the following Lie-algebraic cocycle relation for 
the corresponding anomaly $\Delta$ of the AMWI (i.e., $\Delta$ is obtained from $\zeta$ by \eqref{eq:Delta-X}):
\begin{equation}
\begin{split}\label{eq:WZ}
    \Delta([X,Y])(F)=&
    -[\Delta X,\Delta Y]_{\LieRc}(F) +(\partial_X\Delta Y)(F)-(\partial_Y\Delta X)(F)\\
\end{split}
\end{equation}
for $X,Y\in\LieGc$.
\end{theorem}
It is instructive to see how the seemingly different derivations in sections 4 and 5 lead to the same consistency relation for $\Delta$. 
In the preceding section we solely used the definition of $\Delta$ in terms of the AMWI \eqref{eq:AMWI-q}; 
here we solely used the expression of $\Delta$ in terms of $\zeta$ \eqref{eq:Delta-X}.

\section{Infinitesimal cocycle condition from the nilpotency of the BV operator and relation to $L_\infty$-algebras}\label{sec:BV}
In this section, we will derive the infinitesimal cocycle condition \eqref{eq:WZ} within the BV formalism. The crucial insight is that the infinitesimal renormalisation group transformation $\Delta X$, applied to a local functional $F$, can in fact be identified with the renormalized BV Laplacian $\Lap_F$ for the interaction $F$, applied to the vector field $\partial_X$, 
\be\label{eq:DeltaX-LapF}
\Delta X(F)=i\,\Lap_F(\partial_X)\ ,\ X\in\LieGc,\ F\in\Floc(M)\ ,
\ee
where $\partial_X$ is the vector field on $\Ec(M,\RR^n)$ induced by $X$ -- see \cite{FredenhagenR13}. 
We rederive this result below in \eqref{eq:LapF-A'} and \eqref{eq:DeltaX-A}.
The operator $\Lap_F$ can be expressed by means of a generalization of the AMWI (see \eqref{eq:LapF-AMWI}), 
so the derivation of the anomaly consistency condition \eqref{eq:WZ-0} given in the current section is essentially equivalent to the previous 
derivation in Sect.~\ref{sec:infinitesimal-cc}. Phrasing it in terms of the BV language, however, is important for showcasing the underlying algebraic structures naturally associated with the space of multivector fields and allows us to make connection with the literature, in particular \cite{Hollands08,Fro}.

The BV formalism gives an interpretation of the space of functionals on the solutions to equations of motion in the language of homological algebra. It also allows one to reformulate the deformation of the pointwise product of functionals into their time-ordered product as a deformation of a certain differential. 
We will review the main results concerning the application of this formalism in pAQFT, while presenting it in a slightly different way, emphasizing the geometric interpretation.

We start by observing that functions and vector fields on $\Ec$ can be considered as elements of a larger algebra, namely, the graded commutative algebra $\BV(M)$ of multivector fields. This space has an interpretation as a space of functions on a graded manifold, namely the (-1)-shifted cotangent bundle $T^*[-1]\Ec(M,\RR^n)$ over the configuration space. The latter is just identified with $\overline{\Ec}(M)\equiv\Ec(M,\RR^n)\oplus\Ec_{\mathrm{dens}}(M,\RR^n)[-1]$, where the number $-1$ in square brackets indicates that elements of this space are to be seen as odd variables of degree $-1$. 

More concretely, we identify $\frac{\delta}{\delta\phi}$ with degree $-1$ generators $\Phi^{\ddagger}$, called \emph{antifields}, so
\be
\Phi_r^{\ddagger}(x)[dF[\phi]]=\frac{\delta F}{\delta\phi_r(x)}[\phi]\ ,\ F\in\Fc(M)\ 
\ee
and the elements $\F\in\BV(M)$ are of the form
\be
\F=\sum_{n,m}\langle f_{nm},\Phi^{\otimes n}\otimes(\Phi^{\ddagger})^{\otimes m}\rangle
\ee 
where the compactly supported distributions $f_{n,m}$ are symmetric in the first $n$ and antisymmetric in the last $m$ arguments.
If $f_{nm}=0$ for $m\neq 1$, the element $\F$
can be identified with a vector field on $\Ec(M,\RR^n)$. 
The wave-front set conditions on $f$ are the same as for the distributions characterizing elements of $\Fc(M)$. Analogously as for the functionals on the original configuration space, we introduce the spaces $\BVloc(M)$, $\BVnloc{n}(M)$ and $\BVmloc(M)$.
 
The algebra $\BVmloc(M)$ is equipped with a graded Poisson bracket, the \textit{Schouten bracket}, also known as \textit{antibracket}.  
For a functional $F\in\Fmloc(M)$ and a vector field $\X\in\BVmloc(M)$ it is given by the action of the vector field on the functional as a derivation: $\{\X,F\}\doteq \X F$. For two vector fields we have $\{\X,\Y\}=[\X,\Y]$, \ie the Lie bracket of vector fields, and for general elements $\F,\Gc\in\BVmloc(M)$ we invoke the graded Leibniz rule. 

In this notation, the antibracket takes the form:
\be\label{eq:antibracket}
\{\F,\Gc\}=\left\langle\frac{\delta^r\F}{\delta\phi},\frac{\delta^l \Gc}{\delta \phi^{\ddagger}}\right\rangle-\left\langle\frac{\delta^r\F}{\delta\phi^{\ddagger}},\frac{\delta^l \Gc}{\delta \phi}\right\rangle
\ee
where $\delta^r$ and $\delta^l$ signify right and left derivatives, respectively. We will use the convention that if no superscript is present, then the derivative is to be understood as the left derivative (see \cite{FredenhagenR13} for more detail).

The physical information about the equations of motion and symmetries of the classical theory with Lagrangian $\mathcal{L}+\mathcal{F}$ (both now map into $\BVmloc(M)$), is encoded in the \textit{classical BV operator}. For a compactly supported multivector field $\X\in\BVmloc(M)$, we define
\be\label{eq:sF}
s_{\F} \X \doteq \{\X,\Lc(f)+\F(f)\}\,,
\ee
where $f\equiv 1$ on $\supp \X$. To simplify the notation, we will often just write $\{\X,\Lc+\F\}$, unless we want to indicate a particular choice of the test function. To avoid signs, in what follows we will assume our Lagrangians and observables to be of the form $\X=\sum \eta_i\X_i$, $\X_i\in \BVloc(M)$ and $\eta_i$ elements of a multiplier Grassmann algebra such that $\X$ is 
even\footnote{For the use of external multipliers from Grassmann algebras (the $\eta$-trick \cite{D19}) see \eg \cite{BDFR20}. In particular note that, if $\eta$ is odd, then $\eta\phi^\ddagger=-\phi^\ddagger\eta$.}.

In the absence of interaction, we consider the free classical BV operator:
\be\label{eq:sF0}
s_0 \X \doteq \{\X,\Lc\}\,,
\ee

We say that the theory satisfies the \emph{classical master equation} (CME) if
\[
\{\Lc(f)+\F(f),\Lc(f)+\F(f)\}=0\,.
\]
In order to ensure that this equation holds exactly (rather than up to terms supported within the suport of $df$), it might be necessary to make some particular choices of the smearing function $f$ (see e.g. \cite{BFRej13}). As a consequence of this equation and the graded Jacobi identity for the Schouten bracket, $s_{\F}$ is a differential,
\be\label{eq:sF2=0}
(s_{\F})^2(\X)=\{\{\X,\Lc+\F\},\Lc+\F\}=\frac12\{\X,\{\Lc+\F,\Lc+\F\}\}=0\ ,
\ee
The space of on-shell functionals is encoded in the 0-th cohomology of the differential $s_{\F}$, and the first cohomology gives the space of non-trivial (i.e. not vanishing on-shell) symmetries.

Quantising the theory corresponds to the deformation of the BV differential. For the free theory, we define
\be\label{eq:hat-s0}
\hat{s}_0\doteq T^{-1}\circ s_0\circ T
\ee
where the time ordering operator $T$ is extended to multilocal functionals of fields and antifields, $\F\in\BVmloc(M)$, by treating antifields as classical sources (they are not affected by $T$ and the $\star$ product defined in \eqref{eq:star} is extended to antifields as the pointwise product). 
Since $T$ is linear and $s_0$ is linear and nilpotent, $\hat{s}_0$ is linear and nilpotent.

Operator $\hat{s}_0$ has a particularly nice expression on the space of
polynomials of linear local functionals. These are called \emph{regular functionals} $\BV_{\mathrm{reg}}(M)\subset\BVmloc(M)$ and their functional derivatives at all orders at all points are smooth. 
We compute
\be
\hat{s}_0(\F)=(s_0-i\Lap)(\F)
\ee 
with the BV Laplacian
\be\label{eq:BVLapnr}
\Lap=\int_M\frac{\delta^2}{\delta\phi(x)\delta\phi^{\ddagger}(x)}\ .
\ee
Following \cite{FredenhagenR13}, we define the 
renormalized BV Laplacian in the absence of interaction by
\be
\Lap_0\doteq i(\hat{s}_0-s_0)\ .
\ee

In the presence of an interaction $\F\in \BVloc(M)$, there are two natural ways to define the \emph{quantum BV operator}. Assume first 
that $\F,\X\in\BV_{\mathrm{reg}}(M)$.%
\footnote{This assumption is valid until \eqref{eq:QME1}.}
On one hand we define
\be\label{eq:shatF}
\hat{s}_{\F}(\X)\doteq e^{-i\F}\hat{s}_0(e^{i\F}\X)\ ,
\ee
On the other hand, we set:
\be\label{eq:stildF}
\tilde{s}_{\F}(\X)\doteq R_{\F}^{-1}\circ s_0\circ R_{\F} (\X)\ ,
\ee
where
\[
R_{\F}(\X)\doteq (Te^{i \F})^{-1}\star T(e^{i \F}\X)\,
\]
is the retarded M{\o}ller operator, which maps functionals to interacting quantum observables. This differential is more natural than $\hat{s}_{\F}$, since it is the obvious deformation of $\hat{s}_0$ when passing from the free quantum theory with the star product $\star$ to the interacting quantum theory with the star product on $\BV_{\mathrm{reg}}(M)[[\hbar]]$, defined by:
\[
\X\star_{\F}\Y\doteq R^{-1}_{\F}(R_{\F}\X\star R_{\F}\Y)\,.
\]
Both operators are nilpotent and  $\tilde{s}_{\F}$ is in addition a derivation with respect to $\star_{\F}$
(since $s_0$ is a derivation w.r.t. $\star$). A short calculation, relying on the fact that $s_0$ is a derivation with respect to $\star$, shows that the relation between the two operators is given by:
\begin{align*}
\tilde{s}_{\F}(\X)&=R_{\F}^{-1}\circ s_0((Te^{i \F})^{-1}\star T(e^{i \F}\X))\\
&=R_{\F}^{-1}\bigl(s_0((Te^{i \F})^{-1})\star T(e^{i \F}\X)+(Te^{i \F})^{-1}\star s_0(T(e^{i \F}\X))\bigr)\\
&=R_{\F}^{-1}\bigl(-(Te^{i \F})^{-1}\star s_0(Te^{i \F})\star (Te^{i \F})^{-1}\star T(e^{i \F}\X)+(Te^{i \F})^{-1}\star T(e^{i \F}\hat{s}_{\F}(\X))\bigr)\\
&=R_{\F}^{-1}\bigl(- R_{\F}\circ \hat{s}_{\F}(1)\star R_{\F}(\X)+R_{\F}(\hat{s}_\F(\X))\bigr)\\
&=\hat{s}_{\F}(\X)-\hat{s}_{\F}(1)\star_{\F} \X\,,    
\end{align*}
so they coincide if $\hat{s}_{\F}(1)=0$, which can also be expressed as
\be\label{eq:QME}
s_0(Te^{i\F})=0\,,
\ee
and is the condition that the formal S-matrix is invariant under $s_0$  \cite{FredenhagenR13}. It is equivalent to
\be\label{eq:QMEshort}
\{\Lc,\F\}+\frac{1}{2}\{\F,\F\}-i\Lap(\F)=0\,,
\ee
This equation follows if $\Lc$ and $\Lc+\F$ satisfy the usual \emph{quantum master equation} (QME)\cite{HenneauxT92}, i.e.:
\be\label{eq:QME1}
\frac{1}{2}\{\Lc+\F,\Lc+\F\}-i \Lap(\Lc+\F)=0\,,\quad \frac{1}{2}\{\Lc,\Lc\}-i \Lap(\Lc)=0\,.
\ee
Typically, we assume that $\Lc$ does not depend on antifields, so the latter condition is trivially satisfied. In the following, we will assume that $\Lc$ satisfies both the QME and the CME.

Now we generalize the discussion above to the situation where $\F$ is local,
which amounts to renormalisation. The effects of renormalization can be understood by the 
AMWI which was extended to local multivector fields by Hollands and takes the form
\cite[Prop. 3]{Hollands08}:
\be\label{eq:AMWI-general}
s_0(Te^{i\F})=i T\Bigl( \bigl(\tfrac{1}{2}\{\Lc+\F,\Lc+\F\}+A(\F)\bigr)e^{i\F}\Bigr)\,,
\ee
where $A$ characterizes the anomalies and replaces the ill-defined BV Laplacian $-i\Lap$ in equation \eqref{eq:QMEshort} (one assumes that $\Lc$ satisfies the CME). It is of the form 
\be
A(\F)=\sum_{n=0}^{\infty}\frac{1}{n!}A_{n}(\F^n)\ , \ \F\in\BVnloc{1}(M) \text{ even }\,,
\ee 
where $A_{n}:\BVnloc{n}(M)\to\BVloc(M)$ are linear maps, which reduce the antifield number by 1, hence, $A(\F)$ is odd.
In particular, for $F\in\Floc(M)$ we see that $A(F)=0$.

The renormalized quantum BV operator $\hat{s}_{\F}$ is still given by \eqref{eq:shatF}, so that the generalized AMWI \eqref{eq:AMWI-general} can
equivalently be written as
\be\label{eq:AMWIvf3}
\hat{s}_\F(e^{i\X})=i\,e^{i\X}\left(\frac{1}{2}\{\X+\F+\Lc,\X+\F+\Lc\}+ A(\F+\X)\right)\ .
\ee
We introduce the interaction-dependent BV Laplacian by
\be\label{eq:LapF-def}
\Lap_{\F} \doteq i(\hat{s}_{\F}-s_{\F})\ .
\ee
On regular functionals, $\Lap_{\F}=\Lap_0=\Lap$,
but due to renormalization, the operators differ in general. 
Since $\hat{s}_0$ is linear, also $\hat{s}_{\F}$ and $\Lap_{\F}$ are linear.
From \eqref{eq:sF2=0} and \eqref{eq:hat-s0} we immediately see that $(\hat s_0)^2=0$; hence, by \eqref{eq:shatF}, also $(\hat s_{\F})^2=0$.

The renormalized version of the QME  is again 
\eqref{eq:QME}. By using \eqref{eq:AMWI-general}  it is equivalent to
\be\label{QMEDelta1}
-\Lap_{\F}(1)=\frac{1}{2}\{\Lc+\F,\Lc+\F\}+A(\F)=0\ .
\ee
For $\F$ satisfying QME \eqref{QMEDelta1}, the AMWI \eqref{eq:AMWIvf3} simplifies to
\be\label{eq:AMWIvf}
\hat{s}_\F(e^{i\X})=i\,e^{i\X}\left(\{\X,\Lc+\F\}+\frac{1}{2}\{\X,\X\}+ A(\F+\X)-A(\F)\right)\,,
\ee

The relation between $A$ and $\Lap_{\F}$ is obtained from \eqref{eq:LapF-def} and \eqref{eq:AMWIvf} by using that $s_{\F}$ is a derivation:
\be\label{eq:LapF-AMWI}
-i\Lap_{\F}(e^{i\X})=i\,e^{i\X}\left(\frac{1}{2}\{\X,\X\}+ A(\F+\X)- A(\F)\right)\ .
\ee
Taking into account that $\Lap_{\F}$ is linear, this formula implies
\be\label{eq:LapF-A'}
\Lap_{\F}(\X)=-i\frac{d}{d\lambda}\Big\vert_{\lambda=0}\Lap_{\F}(e^{i\lambda\X})=i\langle A'(\F),\X\rangle\ ,
\ee 
for $\F$ satisfying QME.

Note that for $\X=\partial_X\eta$ (where $\eta$ is Grassmann odd in order that $\X$ is even) 
the original AMWI \eqref{eq:anomMWI} (or \eqref{eq:AMWI-q})
is obtained from the generalized AMWI \eqref{eq:AMWIvf} as the coefficient of $\eta$.
Namely we get
\be
T\bigl(e^{iF}\hat{s}_F(\partial_X)\bigr)[\phi]=s_0\circ T\bigl(e^{iF}\partial_X\bigr)[\phi]=T\bigl(e^{iF}\partial_X\langle\Phi,q\rangle\bigr)[\phi] 
\ee
with  $F\in\Fc_{\loc}$, $L$ used in the definition of the free theory not depending on antifields and $\phi$ satisfying $q=\frac{\partial L}{\partial\phi}[\phi]$. This is similar to the result of \cite{FredenhagenR13}, with the difference that here we introduced the external source $q$. In the last formula on the right-hand side, $q$ can be pulled out from under the time-ordering operator and after one sets  $q=\frac{\partial L}{\partial\phi}[\phi]$, one obtains the same relation between $\hat{s}_F$ and $s_0$ as in \cite{FredenhagenR13}.

Now, applying AMWI \eqref{eq:AMWI-q} to the right-hand side, we obtain
\be
T\bigl(e^{iF}\hat{s}_F(\partial_X)\bigr)[\phi]=
T\Bigl(e^{iF}\bigl(\partial_X(L+F)-\Delta X(F)\bigr)\Bigl)[\phi]\ .
\ee 
Since 
$A(F)=0$ this coincides with the corresponding term for the right hand side of \eqref{eq:AMWIvf}, i.e.,
$\tfrac{d}{d\lambda}\vert_{\lambda=0}T\bigl(e^{i(F+\lambda\X)}\bigl(\lambda\{\X,L+F\}+\frac{\lambda^2}{2}\{\X,\X\}+A(F+\lambda\X)\bigr)\bigr)$,
with the identification $\langle A'(F),\partial_X\eta\rangle=-\Delta X(F)\,\eta$, that is,
\be\label{eq:DeltaX-A}
\Delta X(F)=-\frac{d^r}{d\eta}\Big\vert_{\eta=0}A(F+\partial_X\eta)\doteq -\langle A'(F),\partial_X\rangle\ ,
\ee
hence by \eqref{eq:LapF-A'} we indeed obtain the announced relation \eqref{eq:DeltaX-LapF} between $\Delta X$ and $\Lap_F$.

\medskip
Within the BV formalism, the anomaly consistency condition is a consequence of the nilpotency of the BV operator, in particular, in the context of pAQFT, this was discussed already in the work of Hollands \cite{Hollands08}. It is shown in \cite[Prop.5]{Hollands08}
that the nilpotency of $\hat{s}_0$, i.e. $\hat{s}_0^2=0$, (or, as observed in \cite{FredenhagenR13}, the nilpotency of $\hat{s}_{\F}$) induces a consistency condition for the anomaly term in \eqref{eq:AMWIvf3}. We recall this result in Proposition~\ref{prop:CCBV} and provide an alternative (shorter) proof  using a result of Fr\"ob \cite{Fro}, which also highlights the $L_{\infty}$-structure underlying the BV quantisation.

It was shown by Fr\"ob \cite{Fro} that 
there is an $L_{\infty}$-structure on $\BVnloc{1}(M)$ coming from the AMWI \eqref{eq:AMWIvf3}. The 
brackets $[\bullet,\dots,\bullet]^{\F}_n:\BVnloc{1}(M)^n\to\BVnloc{1}(M)$ are linear and graded symmetric maps,  
given in terms of the generating function (for even $\X$) by
\be\label{eq:def-brackets}
[e^{i\X}]^{\F}\equiv\sum_{n=0}^{\infty}\frac{i^n}{n!}[\X,\dots,\X]^{\F}_n\doteq e^{-i\X}\hat{s}_{\F}(e^{i\X})\ .
\ee
Note that $[e^{i\X}]^{\F}$ is odd. Obviously, with this definition, the AMWI \eqref{eq:AMWIvf3} can be written as
\be\label{eq:AMWI-brackets}
[e^{i\X}]^{\F}=i\bigl(\tfrac{1}2\{\Lc+\F+\X,\Lc+\F+\X\}+A(\F+\X)\bigr)\ .
\ee
Crucially, we verify here (streamlining the argument of \cite{Fro}) that the brackets defined by the formula~\eqref{eq:def-brackets} satisfy the generalized Jacobi identity.
\begin{proposition}
The brackets defined by \eqref{eq:def-brackets} satisfy the generalized Jacobi identity:
\be\label{eq:Jacobi-brackets}
[e^{i\X},[e^{i\X}]^{\F}]^{\F}=0\,.
\ee
\end{proposition}
\begin{proof}
The result follows directly from the nilpotency of $\hat{s}_{\F}$
and the fact that $\hat{s}_{\F}(e^{i\X})$ is odd.
To see this, first note that, for $\Gc\in\BVnloc{1}(M)$ even, we obtain
\be
[e^{i\X},\Gc]^{\F}=\frac{d}{i\,d\lambda}\Big\vert_{\lambda=0}[e^{i(\X+\lambda\Gc)}]^{\F}\overset{\eqref{eq:def-brackets}}{=}
-e^{-i\X}\Gc\,\hat{s}_{\F}(e^{i\X})+e^{-i\X}\hat{s}_{\F}(e^{i\X}\Gc).
\ee
Inserting $\Gc=\eta\,[e^{i\X}]^{\F}=\eta \,e^{-i\X}\hat{s}_{\F}(e^{i\X})$ (with $\eta$ an odd Grassmann variable) and omitting in the resulting formula 
the factor $\eta$, we get
\[
[e^{i\X},[e^{i\X}]^{\F}]^{\F}= -e^{-i2\X}\bigl(\hat{s}_{\F}(e^{i\X})\bigr)^2+e^{-i\X}\hat{s}_{\F}^2(e^{i\X})=0\ .
\]
\end{proof}
From \eqref{eq:def-brackets} we see that the 0-bracket is
\[
[-]_0^{\F}=\hat{s}_{\F}(1)\,,
\]
so vanishes identically if $\F$ satisfies the QME \eqref{eq:QME} or \eqref{QMEDelta1}, and that the 1-bracket is given by
\be
[\X]_1^{\F}=\hat{s}_{\F}(\X)=s_{\F}(\X)-i\Lap_{\F}(\X)\,.
\ee
From \eqref{eq:AMWI-brackets} we obtain for
the 2-bracket 
\be
[\X,\X]_2^{\F}=-i\bigl(\{\X,\X\}+\langle A''(\F),\X\otimes \X\rangle\bigr)
\ee  
and for the $n$-bracket 
(with $n>2$) 
\be\label{eq:n-bracket}
[\X,\dots,\X]_n^{\F}=(-i)^{n-1}\langle A^{(n)}(\F),\X^{\otimes n}\rangle\ .
\ee
Hence, we have an $L_\infty$ structure, provided $\Lc+\F$ satisfies the QME.
Now we can come back to \cite[Prop.5]{Hollands08}.
\begin{proposition}\label{prop:CCBV}
The anomaly $A(\F)$  defined by the generalized AMWI \eqref{eq:AMWIvf3} satisfies the relation
\be\label{eq:CCBV}
0=\{\Lc+\F,A(\F)\}+\langle A'(\F),\left(\tfrac{1}{2}\{\Lc+\F,\Lc+\F\}+ A(\F)\right)\rangle\,.
\ee
\end{proposition}
\begin{proof}
We prove this proposition by verifying that the generalized Jacobi identity \eqref{eq:Jacobi-brackets} for the particular value $\X=0$
is precisely the consistency condition \eqref{eq:CCBV} (which is not surprising
since both rely on $\hat{s}_{\F}^2=0$). To verify this, we use the fact that
\be\label{eq:aux}
    [1,\Gc]^{\F}=\frac{d^r}{i\,d\lambda}\Big\vert_{\lambda=0}[e^{i\Gc\lambda}]^{\F}
    \overset{\eqref{eq:AMWI-brackets}}{=}\{\Lc+\F,\Gc\}+\langle A'(\F),\Gc\rangle\ 
\ee
for $\Gc\in\BVnloc{1}(M)$ odd and $\lambda$ an odd Grassmann parameter,
and applying again the AMWI \eqref{eq:AMWI-brackets}, we obtain
\be
\begin{split}\label{eq:CCBV-proof}
&0\, =-i\,[1,[1]^{\F}]^{\F}\overset{\eqref{eq:AMWI-brackets}}{=}
[1\,,\bigl(\tfrac{1}{2}\{\Lc+\F,\Lc+\F\}+ A(\F)\bigr)]^{\F}\nonumber\\
&\overset{\eqref{eq:aux}}{=}\{\Lc+\F\,,\,A(\F)\}
+\langle A'(\F),\bigl(\tfrac{1}{2}\{\Lc+\F,\Lc+\F\}+A(\F)\bigr)\rangle\,,
\end{split}
\ee
where we also used the graded Jacobi identity for the antibracket; hence we arrive at \eqref{eq:CCBV}.
\end{proof}

Note that the vector fields $\partial_X$, $X\in \LieGc$ are of at most first order in $\phi$. 
For these vector fields we have $\langle A''(F),\partial_X\otimes \partial_Y\rangle=0$
(see \cite{Fro} for a related result).

To show this we start with the following lemma.
\begin{lemma}\label{Tder}
Let $\F\in T(\BVmloc(M))$ and $\X\in\BVnloc{1}(M)$ depend at most linearly on $\phi$. Then
\be
s_0(\F\cdot_T\X)=s_0(\F)\cdot_T\X+\F\cdot_{T_1} s_0(\X)+i\{\F,\X\}\ .
\ee
with the first order time ordered product
\be
A\cdot_{T_1} B\doteq A\cdot B+\langle A',E^{\mathrm{F}}B'\rangle\ ,\ A,B\in T(\BVmloc(M))\ .
\ee
\end{lemma}
\begin{proof}
First note that $\F\in T(\BVmloc(M))$ implies $s_0(\F)\in T(\BVmloc(M))$, due to the generalized AMWI \eqref{eq:AMWI-general},
and that $T\X=\X$, since $\X$ is at most linear in $\phi$.
We use the off-shell field equation \eqref{eq:T-field-eq} and the fact that $s_0$ is
 a derivation for the pointwise product. The derivative with respect to $\phi$ is denoted by ${}'$. We have
 \be\label{eq:s_0(FTX)}
\begin{split}
&s_0(\F\cdot_T \X)
=s_0(\F\cdot \X+\langle \F',E^{\mathrm{F}}\X'\rangle)=s_0(\F)\cdot \X+\F\cdot s_0(\X)+\langle s_0(\F'),E^{\mathrm{F}}\X'\rangle+\langle\F',E^{\mathrm{F}}s_0(\X')\rangle\\
&=s_0(\F)\cdot_T\X+\F\cdot_{T_1} s_0(\X)+\langle s_0(\F')-s_0(\F)',E^{\mathrm{F}}\X'\rangle+\langle\F',E^{\mathrm{F}}\bigl(s_0(\X')-s_0(\X)'\bigr)\rangle
\end{split}
\ee
However, with $s_0=-\langle\frac{\delta L}{\delta\phi},\frac{\delta}{\delta\phi^{\ddagger}}\rangle$ and
\be\label{eq:Ef->delta}
\int dy\,\,E^{\mathrm{F}}(z,y)\,\frac{\delta^2 L}{\delta\phi(y)\delta\phi(x)}=i\delta(z,x)\ ,
\ee
it holds for any $\G\in \BV(M)$
\be
E^{\mathrm{F}}(s_0(\G')-s_0(\G)')=i\frac{\delta\G}{\delta\phi^{\ddagger}}\ .
\ee
Hence the last two terms in \eqref{eq:s_0(FTX)} form the antibracket $i\{\F,\X\}$ and 
we obtain the statement in the lemma.
\end{proof}

Next, we show that for elements of first order in $\phi$, $\Lap_{\F}$ acts as the unrenormalized BV Laplacian, up to an extra term of the form $\Lap_\F(1)\X\Y$.

\begin{proposition}\label{prop:BV-Laplace}
Let $\X,\Y\in\BVnloc{1}(M)$ be of first order in $\phi$ and even (we multiply the usual vector fields with Grassman generators
-- the $\eta$-trick), $\F\in\BVnloc{1}(M)$ even and $L$ independent of antifields. Then
\be\label{eq:DeltaF(XY)}
\Lap_\F(\X\Y)=(\Lap_\F\X)\Y+\X(\Lap_\F\Y)+\{\X,\Y\}-\Lap_\F(1)\X\Y\,.
\ee
\end{proposition}
\begin{proof}
Since $\Lap_\F=i(\hat{s}_\F-s_\F)$ and $s_\F$ is a derivation, the statement is equivalent to
\be
\hat{s}_\F(\X\Y)=\hat{s}_\F(\X)\Y+\X\hat{s}_\F(\Y)-i\,\{\X,\Y\}-\hat{s}_\F(1)\X\Y\ .
\ee
Taking into account that $\hat{s}_\F(\X)=e^{-i\F}T^{-1}s_0(Te^{i\F}\X)$ this is equivalent to
\be
s_0(Te^{i\F}\X\Y)=s_0(Te^{i\F}\X)\cdot_T\Y+\X\cdot_Ts_0(Te^{i\F}\Y)-i\,T(e^{i\F}\{\X,\Y\})-s_0(Te^{i\F})\T \X\T \Y\ ,
\ee
by applying \eqref{eq:cdotT} and that $T^{-1}\X=\X$ (and similarly for $\Y$).
 Since also $T^{-1}\{\X,\Y\}=\{\X,\Y\}$, it remains to show that for $\G\in T(\BVmloc(M))$ it holds that
 \be
s_0(\G\cdot_T\X\cdot_T\Y)-s_0(\G\cdot_T\X)\cdot_T\Y-s_0(\G\cdot_T\Y)\cdot_T\X+s_0(\G)\cdot_T\X\cdot_T\Y=i\G\cdot_T\{\X,\Y\}\ .
\ee
For this purpose we use Lemma \ref{Tder}. We get
\be\label{eq:s0-dotT}
\begin{split}
&s_0(\G\cdot_T\X\cdot_T\Y)-s_0(\G\cdot_T\X)\cdot_T\Y-s_0(\G\cdot_T\Y)\cdot_T\X+s_0(\G)\cdot_T\X\cdot_T\Y\\
&=s_0\bigl((\G\cdot_T\X)\cdot_T\Y\bigr)-s_0(\G\cdot_T\X)\cdot_T\Y-(\G\cdot_T\X)\cdot_{T_1} s_0(\Y)\\
&+\bigl(\G\cdot_{T_1} s_0(\Y)+s_0(\G)\cdot_T\Y-s_0(\G\cdot_T\Y)\bigr)\cdot_T\X+Z\\
&=i\{\G\cdot_T\X,\Y\}-i\{\G,\Y\}\cdot_T\X+Z\\
\end{split}
\ee
where
\be
Z\doteq -\X\cdot_{T_1}\bigl(\G\cdot_{T_1} s_0(\Y)\bigr)+\bigl(\X\cdot_{T_1}\G\bigr)\cdot_{T_1} s_0(\Y)=
-\langle s_0(\Y)'',E^{\mathrm{F}}\X'\otimes E^{\mathrm{F}}\G'\rangle\ .
\ee
But $Z$ is just the correction to the derivation property of the antibracket with respect to the time ordered product,
that is, the r.h.s.~of \eqref{eq:s0-dotT} is indeed equal to $i\G\cdot_T\{\X,\Y\}$.
To wit, we have
\be
\begin{split}
&\{\G\cdot_T\X,\Y\}-\{\G,\Y\}\cdot_T\X-\G\cdot_T\{\X,\Y\}\\
&=\{\langle \G',E^{\mathrm{F}}\X'\rangle,\Y\}-\langle\{\G,\Y\}',E^{\mathrm{F}}\X'\rangle-\langle\G',E^{\mathrm{F}}\{\X,\Y\}'\rangle\\
&=\langle\{\G',\Y\},E^{\mathrm{F}}\X'\rangle+\langle\{\X',\Y\},E^{\mathrm{F}}\G'\rangle-\langle\{\G,\Y\}',E^{\mathrm{F}}\X'\rangle-\langle\G',E^{\mathrm{F}}\{\X,\Y\}'\rangle\\
&=-\langle\{\G,\Y'\},E^{\mathrm{F}}\X'\rangle-\langle\{\X,\Y'\},E^{\mathrm{F}}\G'\rangle\\
&=-\int dxdydz\frac{\delta^2\Y}{\delta\phi(x)\delta\phi^{\ddagger}(y)}E^{\mathrm{F}}(x,z)\Bigl(\frac{\delta\X}{\delta\phi(z)}\frac{\delta\G}{\delta\phi(y)}+\frac{\delta\X}{\delta\phi(y)}\frac{\delta\G}{\delta\phi(z)}\Bigr)\ ,
\end{split}
\ee
and, by using \eqref{eq:Ef->delta}, this coincides with $iZ$ since
\be
s_0(\Y)''(x,y)=-\int dz \Bigl(\frac{\delta^2 L}{\delta\phi(x)\delta\phi(z)}\frac{\delta^2\Y}{\delta\phi(y)\delta\phi^{\ddagger}(z)}+
\frac{\delta^2 L}{\delta\phi(y)\delta\phi(z)}\frac{\delta^2\Y}{\delta\phi(x)\delta\phi^{\ddagger}(z)}\Bigr)\ .
\ee
\end{proof}

\begin{remark}
    Note that if, in addition, the quantum master equation \eqref{eq:QME} holds, then $\Lap_\F(1)=0$ and we obtain the more familiar relation:
\be\label{eq:DeltaF(XY)b}
\Lap_\F(\X\Y)=(\Lap_\F\X)\Y+\X(\Lap_\F\Y)+\{\X,\Y\}\,.
\ee
\end{remark}
The result on $A''$ follows now directly from Proposition \ref{prop:BV-Laplace}. 
\begin{proposition}\label{prop:secondorder}
Let $A$ be the anomaly appearing in the AMWI. Let $\F\in\BVnloc{1}(M)$ even, $L$ independent of antifields and let $\X,\Y\in\BVnloc{1}(M)$ be at most linear in $\phi$.
Then 
\be
\langle A''(\F),\X\otimes\Y\rangle=0\ .
\ee
\end{proposition}
\begin{proof}
Without restriction of generality we may assume that both $\X$ and $\Y$ are even (by using the $\eta$-trick).
Let $\lambda,\mu\in\RR$. 
From the generalized AMWI \eqref{eq:AMWIvf3} we have that:
\begin{align}\label{eq:aux:FXY}
&\hat{s}_\F(e^{i(\lambda \X+\mu \Y)})\equiv (s_\F-i\Lap_\F)(e^{i(\lambda \X+\mu \Y)})\\
&\overset{\eqref{eq:AMWIvf3}}{=}
i\,e^{i(\lambda \X+\mu \Y)} \left(\frac{1}{2}\{\lambda \X+\mu \Y+L+\F, \lambda \X+\mu \Y+L+\F\}+ A(\F+\lambda \X+\mu \Y)\right) 
\nonumber
\end{align}
Selecting the terms proportional to $\lambda\mu$ and using the expression \eqref{QMEDelta1} for $\Lap_{\F}(1)$, 
we obtain
\begin{align}
(s_\F  -i\Lap_\F)(\X\Y)
&=\X\bigl(\{\Y,L+\F\}-i\Lap_\F(\Y)\bigr)+\Y\bigl(\{\X,L+\F\}-i\Lap_\F(\X)\bigr)\nonumber\\
&\quad -i\{\X,\Y\}+i\X\Y \Lap_{\F}(1) 
-i\langle A''(\F),\X\otimes\Y\rangle\,,\nonumber
\end{align}
by using the analog of \eqref{eq:LapF-A'} for $\F\in\BVnloc{1}(M)$.
The statement follows from the derivation property of $s_\F$ and Proposition~\ref{prop:BV-Laplace}.
\end{proof}

For $F\in\Fnloc{1}(M)$,  one obtains a map
\be
\LieGc\times\Floc(M)\ni(X,F)\mapsto s_F(\partial_X)-i\Lap_F(\partial_X)=\partial_X(F+L)-\Delta X(F)\,,
\ee
which coincides with the action previously constructed in \eqref{eq:Lie-rep}. The fact that it is an action was derived from the cocycle relation \eqref{eq:WZ} for $X\mapsto \Delta X$ as a consequence of the cocycle relation for the anomaly map $\zeta$ in the UAMWI. 

Actually, the cocycle relation for $\Delta$ in the form of the equivalent consistency relation \eqref{eq:WZ-0} derives directly from 
the BV-consistency condition \eqref{eq:CCBV}.

\begin{proposition}\label{prop:BV->WZ}
The BV consistency relation \eqref{eq:CCBV} implies the extended Wess-Zumino consistency relation \eqref{eq:WZ-0}. 
\end{proposition}
\begin{proof}
Let $L$ be independent of antifields. We  insert
 $\F=F+\partial_{X_1}\eta_1+\partial_{X_2}\eta_2$ into the BV consistency relation \eqref{eq:CCBV},
 where $F\in\Floc(M)$ and $\eta_1,\eta_2$ are Grassmann generators. We use the fact that 
 $\langle A''(F),\partial_{X_1}\otimes\partial_{X_2}\rangle=0$.
 Since $A(F)=0$ we obtain the following finite Taylor expansion in $\eta_1,\eta_2$:
 \be\label{eq:A(F)}
 A(\F)=-\Delta X_1(F)\,\eta_1-\Delta X_2(F)\,\eta_2\,,
 \ee
 where we also used \eqref{eq:DeltaX-A}. In particular note that $\frac{\delta A(\F)}{\delta\phi^\ddagger}=0$.
 With that we obtain
 \begin{align}\label{eq:part1}
 \{L+\F,A(\F)\}&=-\{(\partial_{X_1}\eta_1+\partial_{X_2}\eta_2),\bigl(\Delta X_1(F)\,\eta_1+\Delta X_2(F)\,\eta_2\bigr)\}\nonumber\\
 &=\bigl(-\partial_{X_1}\Delta X_2(F)+\partial_{X_2}\Delta X_1(F)\bigr)\,\eta_1\eta_2\ .
 \end{align}
Note that
 \begin{align}
 \langle A'(\F),(G+\partial_Z\eta)\rangle &=\frac{d}{d\tau}\Big\vert_{\tau=0}A\bigl(\F+\tau (G+\partial_Z\eta)\bigr)\\
 &=-\frac{d}{d\tau}\Big\vert_{\tau=0}
 \Bigl(\Delta X_1(F+\tau G)\,\eta_1+\Delta X_2(F+\tau G)\,\eta_2+\tau\Delta Z(F+\tau G)\,\eta\Bigr)\nonumber\\
 &=-\langle(\Delta X_1)'(F),G\rangle\eta_1-\langle(\Delta X_2)'(F),G\rangle\eta_2 -\Delta Z(F)\,\eta\ ,
 \end{align}
where $G\in\Floc(M)$, $Z\in\LieGc$ and $\eta$ is another Grassmann generator. Hence, using \eqref{eq:bracket}, 
we obtain
 \begin{align}\label{eq:part2}
 \langle A'(\F),\left(\tfrac{1}{2}\{L+\F,L+\F\}\right)\rangle &=
 \langle A'(\F),\left((\partial_{X_1}\eta_1+\partial_{X_2}\eta_2)(L+F)-\partial_{[X_1,X_2]}\,\eta_1\eta_2\right)\rangle
 \nonumber\\
 =\Bigl(-\langle(\Delta X_1)'(F),&\ \partial_{X_2}(L+F)\rangle+\langle(\Delta X_2)'(F),\partial_{X_1}(L+F)\rangle+\Delta[X_1,X_2](F)\Bigr)\eta_1\eta_2\  
 \end{align}
  and
 \begin{align}\label{eq:part3}
 \langle A'(\F),\, A(\F)\rangle&=-\langle A'(\F),\left(\Delta X_1(F)\,\eta_1+\Delta X_2(F)\,\eta_2\right)\rangle
 \nonumber\\
 &=\Bigl(\langle(\Delta X_1)'(F),\Delta X_2(F)\rangle-\langle(\Delta X_2)'(F),\Delta X_1(F)\rangle\Bigr)\eta_1\eta_2\,.
 \end{align}
 Composing \eqref{eq:part1}, \eqref{eq:part2} and \eqref{eq:part3}, we obtain the consistency equation \eqref{eq:WZ-0}.
\end{proof}
\begin{remark}
Note that by tracing back the arguments given in this section, we can see that Proposition~\ref{prop:BV->WZ} essentially states that the generalized Wess-Zumino consistency condition is the consequence of $\hat{s}^2_{\F}=0$. We can compare this with a simple fact  that the nilpotency of the non-renormalized BV Lapalacian $\triangle$ (see \eqref{eq:BVLapnr}) implies an analogous statement for vector fields. Without the loss of generality, we assume $\X$ and $\Y$ to be even (we multiply the usual vector fields with Grassman parameters) $\mathcal{X},\mathcal{Y}\in \BV_{\mathrm{reg}}(M)$ (regular multivector fields):
\[
    0=\triangle^2(\X \Y)=\Lap\bigl((\Lap \X)  \Y+\X(\Lap \Y)+\{\X,\Y\}\bigr)
    =\partial_{\X}(\Lap \Y)+\partial_{\Y}(\Lap \X)+\Lap(\{\X,\Y\})\,
\]
by using that $\Delta$ satisfies a relation analogous to \eqref{eq:DeltaF(XY)},
where $\partial_{\X}$, $\partial_{\Y}$ denotes the natural action of vector fields on functionals as derivations.
\end{remark}

\section{Summary and Outlook}
Symmetries of the classical configuration space are, in general, modified in the quantum theory by anomalies, which have to satisfy certain consistency relations. In different formulations of quantum field theory these consistency relations appear in different ways, and it is not obvious how they are related to each other. We clarified in this paper how the cocycle relation for the anomaly map in a recently proposed nonperturbative characterization of anomalies (the unitary anomalous master Ward identity (UAMWI) in \cite{BDFR21}) is connected with previous consistency relations. As a byproduct this provides new insights into the older formulations. We also clarified the relation between the occurrence of anomalies and the principle of perturbative agreement proposed in \cite{HW04}.

We gave an elementary proof of a consistency relation for the anomaly map $\Delta$ in the anomalous master Ward identity (AMWI) of \cite{Brennecke08}, restricted to infinitesimal symmetries of the configuration space (Thm.~\ref{thm:cc}),
which was originally derived by Hollands \cite{Hollands08} using the antifield formalism. We named it \textit{extended Wess-Zumino condition}, as it can be understood as an extension to non-quadratic interactions of the well-known Wess-Zumino consistency condition \eqref{eq:Cc-G(X)}.

We then showed that our extended Wess-Zumino consistency condition 
can be deduced from the cocycle relation \eqref{eq:cocycle} for the anomaly map 
$\zeta$ occuring in the UAMWI and, hence, describes a Lie algebraic cocycle of the Lie algebra of the group of compactly supported configuration 
space symmetries $\G_c(M)$ with values 
in the Lie algebra of the St\"uckelberg-Petermann renormalization group (Thm.~\ref{thm:Liecc}). 
Conversely, in the framework of perturbation theory, starting with the AMWI one can derive the UAMWI with an anomaly map $\zeta$ fulfilling
the cocycle relation \eqref{eq:cocycle}, see Thm.~\ref{th:AMWI-uAMWI}.

 We also investigated the connection to the BV formalism (as previously studied in \cite{Hollands08,FredenhagenR13,Fro}) and the underlying algebraic structures. In particular we verified that the extended Wess-Zumino consistency 
condition \eqref{eq:WZ-0} can be obtained from the nilpotency of the BV operator $\hat s_{\F}$, by restricting to symmetries $\g\in\G_c(M)$
(Prop.~\ref{prop:BV->WZ}). Our proof starts with the consistency condition \eqref{eq:CCBV} (Prop.~\ref{prop:CCBV}) proven by Hollands \cite{Hollands08},
which can be understood as a particular case of the generalized Jacobi identity for the underlying $L_{\infty}$-algebras 
(see the proof of Prop.~\ref{prop:CCBV}) and relies on $\hat s_{\F}^2=0$.

It is an interesting open problem to find the group-like structure associated to the $L_{\infty}$-structure for more general symmetries and to understand its relation to renormalization. This will be addressed in our future work.

\section*{Acknowledgement}
    The inclusion of a discussion of perturbative agreement and the generalization of Proposition \ref{prop:secondorder} to    interactions, which do not necessarily satisfy the Quantum Master Equation, was suggested to us by one of the referees
    of this paper, a suggestion which is gratefully acknowledged.


\section*{Data availability statement}
Data sharing is not applicable to this article as no new data were created or analyzed in this study.

\appendix

\section{Off shell UAMWI}\label{app:AMWI-uAMWI}
In \cite[Thm.~10.3(i)]{BDFR21} we showed that in formal perturbation theory the UAMWI holds on shell. In this appendix we prove 
by a slightly improved argument that also the off-shell version of the UAMWI \ref{eq:uni-anom-MWI} holds. 
Here ``off-shell'' means that $\phi$ can be arbitrary, but we also introduce external sources $q$ and at the end we set 
$q=\frac{\delta L}{\delta\phi}[\phi]$, so {\it a priori} our expressions are functions of two variables, $\phi$ and $q$ 
and the UAMWI holds on a subspace where the condition $q=\frac{\delta L}{\delta\phi}[\phi]$ is satisfied. 
Crucially, on that subspace the left-hand side of UAMWI proven below does not depend on $q$.
\begin{theorem}\label{th:AMWI-uAMWI}
In formal perturbation theory,  the AMWI implies the off-shell unitary AMWI
\be
S\circ\zeta_\g(F)[\phi]=S\circ \g_{L_q}(F)[\phi]\ ,\ \text{ for }\ \phi\ { solving }\ q= \frac{\delta L}{\delta\phi}[\phi]\ ,
\quad\text{for all}\ F\in\Floc(M)\ ,\, \g\in\G_c(M)\ ,
\ee
with a cocycle $\zeta$ taking values in  $\Rc_c$ and with $\supp\zeta_\g\subset\supp\g$.
\end{theorem}

\begin{proof} 
Since the elements of $\G_c(M)$ have compact support and depend smoothly on $x$, $\G_c(M)$ must be connected. Therefore,
given any $\g\in\G_c(M)$, there exists a smooth curve $\lambda\mapsto\g^{\lambda}\in\G_c(M)$ with $\g^0=\e$ and $\g^1=\g$. Let 
$X^{\lambda}\in\LieGc$ be defined by $\frac{d}{d\lambda}\g^{\lambda}=X^{\lambda}\g^{\lambda}$. 

In the next step, want to find a smooth curve $\lambda\mapsto\zeta_{\g^\lambda}^{-1}\in\Rc_c$ with $\zeta_\e^{-1}=\mathrm{id}$ and 
\be\label{eq:Z(1-lambda)-1}
\frac{d}{d\lambda}S\bigl(\g^{\lambda}_{L_q}\zeta_{\g^\lambda}^{-1}(F)\bigr)[\phi]=0\ \text{with}\ q=\frac{\delta L}{\delta\phi}[\phi]\ .
\ee
Note that inserting $\lambda =0$ and $\lambda =1$ into $S\circ\g^{\lambda}_{L_q}\circ \zeta_{\g^\lambda}^{-1}$, 
we obtain the unitary AMWI \eqref{eq:uni-anom-MWI}.

To search for the desired curve, we will first derive a differential equation that it has to solve.
Note that for $G\in\Floc(M)$, we have
\be
\frac{d}{d\lambda}\g^{\lambda}_{L_q}G=\partial_{X^{\lambda}}\g^{\lambda}_{L_q}G+\partial_{X^{\lambda}}(L_q)\,,
\ee
so by using this result, we perform the differentiation in \eqref{eq:Z(1-lambda)-1} and obtain the condition
\be
S\bigl(\g^{\lambda}_{L_q} \zeta_{\g^\lambda}^{-1}(F)\bigr)\cdot_T\Bigl(\partial_{X^{\lambda}}\g^{\lambda}_{L_q} \zeta_{\g^\lambda}^{-1}(F)+
\partial_{X^{\lambda}}(L_q)+\g^{\lambda}_{\ast}\frac{d}{d\lambda}\zeta_{\g^\lambda}^{-1}(F)\Bigr)[\phi]=0\ \text{ for }q= \frac{\delta L}{\delta\phi}[\phi]\ .
\ee
We insert the anomalous MWI \eqref{eq:AMWI-q} and find 
\be
S\bigl(\g^{\lambda}_{L_q} \zeta_{\g^\lambda}^{-1}(F)\bigr)\cdot_T
\Bigl(\Delta X^{\lambda}(\g_{L_q}^{\lambda}\zeta_{\g^\lambda}^{-1}(F))+\g^{\lambda}_{\ast}\frac{d}{d\lambda}\zeta_{\g^\lambda}^{-1}(F)\Bigr)[\phi]
=0\ \text{for}\ q=\frac{\delta L}{\delta\phi}[\phi]\ .
\ee
On the other hand, $\Delta X^{\lambda}\circ \g_{L_q}^{\lambda}=\Delta X^{\lambda}\circ\g_{L}^{\lambda}$, since 
$\delta_{\g^\lambda}\langle\Phi,q\rangle$ is at most of first order in $\Phi$ and due to the defining property $(iii)$ of $\LieRc$.
We thus get the desired family $\lambda\mapsto \zeta_{\g^\lambda}^{-1}$ as the unique solution of the differential equation
\be\label{eq:DGL-Z-1}
\frac{d}{d\lambda}\zeta_{\g^\lambda}^{-1}=-({\g_{\ast}^{\lambda}})^{-1}\Delta X^{\lambda}\g_L^{\lambda}\zeta_{\g^\lambda}^{-1}\ ,
\ee 
with the initial condition $\zeta_{\g^0}^{-1}=\mathrm{id}$. As explained for the case $q=0$,
\be\label{eq:inR}
(\g_{\ast}^{\lambda})^{-1}\Delta X^{\lambda}\g^{\lambda}_L\in\LieRc\
\ee
holds, hence $\zeta_{\g^\lambda}^{-1}\in\Rc_c$ follows, so in particular $\zeta_\g\in\Rc_c$.
Since the differential equation \eqref{eq:DGL-Z-1} determining $\zeta$
does not contain $q$, we explicitly see that $\zeta$ can be chosen such that it does not depend on $q$ either. Hence, the proof
of $\supp\zeta_\g\subset\supp\g$ can be adopted from the case $q=0$ as it stands.

It remains to show that $\zeta$ satisfies the cocycle identity. 
Applying three times the UAMWI \eqref{eq:uni-anom-MWI} we obtain
\begin{align}\label{eq:cocycle1}
S\circ\zeta_{\g\h}(F)[\phi]= & S\circ(\g\h)_{L_q}(F)[\phi]=S\circ\g_{L_q}\circ\h_{L_q}(F)[\phi]
\nonumber\\
= & S\circ\zeta_\g\circ\h_{L_q}(F)[\phi]=
S\circ\h_{L_q}\circ(\h^{-1}_{L_q}\zeta_\g\h_{L_q})(F)[\phi]\nonumber\\
= &
S\circ\zeta_\h\circ (\h_{L_q}^{-1}\zeta_\g\h_{L_q})(F)[\phi]
\end{align}
for $\phi$ solving $q=\frac{\delta L}{\delta\phi}[\phi]$. Using again that $\delta_\h\langle\Phi,q\rangle$ is at most of first order in $\Phi$,
we conclude that 
\be
\zeta_\g\h_{L_q}(F)=\zeta_\g\bigl(\h_LF-\delta_\h\langle\Phi,q\rangle\bigr)=\zeta_\g\h_L(F)-\delta_\h\langle\Phi,q\rangle\,,
\ee
since $\zeta_g\in\Rc_c$.%
\footnote{The defining property of $Z\in\Rc_c$ corresponding to the renormalization condition ``off shell field equation" for the timeordered product
\eqref{eq:T-field-eq} reads $Z(F+\langle\Phi,q\rangle)=Z(F)+\langle\Phi,q\rangle$, see e.g.~\cite{DF04} and cf.~the defining property $(\mathrm{P}3)$ 
of $\LieRc$. In addition, $Z(F+c)=Z(F)+c$ for $c\in\RR$ follows from the compactness of $\supp Z$, see \cite[Def.~4.2]{BDFR21} 
and cf.~the defining property $(\mathrm{P}5)$ of $\LieRc$.}
With that we obtain
\be
\h_{L_q}^{-1}\zeta_\g\h_{L_q}=\h_L^{-1}\zeta_{\g}\h_L\equiv\zeta_{\g}^{\h}\ ,
\ee
hence the cocycle relation 
\be
S\circ\zeta_{\g\h}(F)=S\circ\zeta_\h\circ(\zeta_\g)^\h(F)
\ee
holds for all field configurations $\phi$.
Since the  \emph{off-shell} S-matrix is injective,
we obtain the cocycle relation for $\zeta$.
\end{proof}

\end{document}